\def\year{2022}\relax
\documentclass[letterpaper]{article} 
\usepackage{aaai22}  
\usepackage{times}  
\usepackage{helvet}  
\usepackage{courier}  
\usepackage[hyphens]{url}  
\usepackage{graphicx} 
\usepackage{natbib}  
\usepackage{caption} 
\DeclareCaptionStyle{ruled}{labelfont=normalfont,labelsep=colon,strut=off} 
\usepackage{algorithm}
\usepackage{algorithmic}

\usepackage{amsmath}
\usepackage{amsthm}
\usepackage{mathrsfs}
\usepackage{dsfont}
\usepackage{amssymb}

\DeclareMathAlphabet{\mathbfit}{OML}{cmm}{b}{it}

\newtheorem{example}{Example}
\newtheorem{thm}{Theorem}
\newtheorem{lem}[thm]{Lemma}
\newtheorem{cor}[thm]{Corollary}
\newtheorem{prop}[thm]{Proposition}

\long\def\comment#1{}
\long\def\appcomment#1{#1}

\usepackage{newfloat}
\usepackage{listings}
\floatstyle{ruled}
\newfloat{listing}{tb}{lst}{}
\floatname{listing}{Listing}
\title{Characterizing the Program Expressive Power of Existential Rule Languages}
\author{Heng Zhang}
\affiliations{
     College of Intelligence and Computing, Tianjin University, China\\
     heng.zhang@tju.edu.cn
}

\usepackage{bibentry}

\begin{document}

\maketitle

\begin{abstract}
Existential rule languages are a family of ontology languages that have been widely used in ontology-mediated query answering (OMQA). However, for most of them, the expressive power of representing domain knowledge for OMQA, known as the program expressive power, is not well-understood yet. In this paper, we establish a number of novel characterizations for the program expressive power of several important existential rule languages, including tuple-generating dependencies (TGDs), linear TGDs, as well as disjunctive TGDs. The characterizations employ natural model-theoretic properties, and automata-theoretic properties sometimes, which thus provide powerful tools for identifying the definability of domain knowledge for OMQA in these languages.
\end{abstract}

\section{Introduction}

Existential rule languages, a.k.a. Datalog$\pm$, had been initially introduced in databases as dependency languages to specify the semantics of data stored in a database~\cite{AbiteboulHV1995}. As one of the most popular dependency languages, tuple-generating dependencies (TGDs) and its extensions, including (disjunctive) embedded dependencies and disjunctive TGDs, had been extensively studied. Recently, these languages have been rediscovered as languages for data exchange~\cite{FaginKMP2005}, data integration~\cite{Lenzerini2002}, ontology reasoning~\cite{CaliGLMP2010} and knowledge graph~\cite{BGPS2017}. 

A major computational task based on existential rule languages is known as ontology-mediated query answering (OMQA), which generalizes the traditional database querying by enriching database with a domain ontology. Unfortunately, even for TGDs, the problem of OMQA was proved to be undecidable~\cite{BV81}. Towards efficient reasoning, many decidable sublanguages have been identified, including linear TGDs and guarded TGDs~\cite{CaliGL2012}, frontier-guarded TGDs~\cite{BagetLMS11}, sticky TGDs~\cite{CaliGP12}, weakly-acyclic TGDs~\cite{FaginKMP2005} and shy programs~\cite{LeoneMTV12}. With these languages, it is thus important to identify their expressive power so that, given an application, we know which language should be used. 

%
%


In OMQA, there have been mainly two lines of research on the language expressive power. The first line of research regards every ontology together with a classical query as a database query, usually called an ontology-mediated query. The main goal of this line is to understand which class of databases can be defined by an ontology-mediated query. We call such kind of expressive power the {\em data expressive power}. 
In contrast, the second line is concerned with which kind of domain knowledge can be expressed in an ontology language; or more formally, which classes of database-query pairs are definable in the language. Expressive power of this kind is known as the {\em program expressive power}, which was first proposed by~\citeauthor{ArenasGP2014}~\shortcite{ArenasGP2014}. 

A number of papers are devoted to characterizing data expressive power of existential rule languages. An incomplete list is as follows: \citeauthor{GRS2014}~\shortcite{GRS2014} proved that weakly (frontier-)guarded TGD queries with stratified negations capture the class of EXPTIME-queries; nearly (frontier-)guarded TGD queries have the same expressive power as Datalog. \citeauthor{RudolphT2015}~\shortcite{RudolphT2015} showed that TGD queries capture the class of recursively enumerable queries closed under homomorphisms. 
\citeauthor{KrotzschR11}~\shortcite{KrotzschR11} identified that jointly acyclic TGD queries have the same expressive power as Datalog, which was later extended to TGD queries with terminating Skolem chase in~\cite{ZhangZY15}. In description logics, \citeauthor{BienvenuCLW14}~\shortcite{BienvenuCLW14} characterized the data expressive power of $\mathcal{ALC}$ and its variants by some interesting complexity classes and fragments of disjunctive Datalog.

Unlike the data expressive power, the program expressive power of existential rule languages is not well-understood yet. \citeauthor{ArenasGP2014}~\shortcite{ArenasGP2014} proved that Datalog is strictly less expressive than warded Datalog$^\exists$, and obtained a similar separation for the variants with stratified negations and negative constraints. \citeauthor{ZhangZY16}~\shortcite{ZhangZY16} proposed a semantic definition for ontologies in OMQA, and proved that disjunctive embedded dependencies (DEDs) capture the class of recursively enumerable OMQA-ontologies. In addition, it is implicit in~\cite{ZhangZY15} that the weakly-acyclic TGDs have the same program expressive power as all its extensions with terminating Skolem chase. This paper continues this line of work and aims at characterizing the program expressive power of several important languages including TGDs,
disjunctive TGDs and linear TGDs.

Our contributions in this paper are threefold. Firstly, we show that the equalities in a finite set of DEDs are removable if, and only if, the OMQA-ontology defined by these DEDs is closed under both database homomorphisms and constant substitutions. Secondly, we prove that, under CQ-answering, every finite set of DTGDs can be translated to an equivalent finite set of TGDs, while the translatability under UCQ-answering is captured by a property called query constructivity.  
Thirdly, we characterize the linear TGD-definability of OMQA-ontologies by data constructivity and the recogniziability of queries by a natural class of tree automata.

\section{Preliminaries}

\subsubsection{Databases and Instances} 

We use a countably infinite set $\Delta$ (resp., $\Delta_{\mathrm{n}}$ and $\Delta_{\mathrm{v}}$) 
of {\em constants} (resp., {\em (labeled) nulls} and {\em variables}), and assume they are pairwise disjoint. Every {\em term} is a constant, a null or a variable. A {\em (relational) schema} $\mathscr{S}$ is a set of {\em relation symbols}, each associated a natural number called the {\em arity}. Every {\em $\mathscr{S}$-atom} is either an equality or a {\em relational atom} built upon terms and a relation symbol in $\mathscr{S}$. A {\em fact} is a variable-free relational atom, and an {\em $\mathscr{S}$-instance} is a set of  $\mathscr{S}$-facts. A {\em database} is a finite instance in which no null occurs. Given an instance $I$, let $adom(I)$ (resp., $term(I)$) denote the set of constants (resp., terms) occurring in $I$. Given a set $A$ of terms, let $I|_A$ be the maximum subset $J$ of $I$ such that $term(J)\subseteq A$.

Let $I$ and $J$ be $\mathscr{S}$-instances, and $C$ a set of constants. A {\em $C$-homomorphism} from $I$ to $J$ is a function $h:adom(I)\rightarrow adom(J)$ such that $h(I)\subseteq J$ and $h(c)=c$ for all constants $c\in C$. If such $h$ exists, we say $I$ is {\em $C$-homomorphic} to $J$, and write $I\rightarrow_C J$. In addition, we write $I\rightarrowtail_C J$ if $h$ is injective. We say $I$ is {\em $C$-isomorphic} to $J$ if there is a bijective $C$-homomorphism $h$ from $I$ to $J$ such that $h(I)=J$. For simplicity, in the above, $C$ could be dropped if it is empty. A {\em substitution} is a partial function from $\Delta_{\mathrm{v}}$ to $\Delta\cup\Delta_{\mathrm{n}}$.

\subsubsection{Queries}

Fix $\mathscr{S}$ as a schema.  Every {\em $\mathscr{S}$-CQ} is a first-order formula of the form $\exists\mathbfit{y}\,\varphi(\mathbfit{x},\mathbfit{y})$ where $\varphi(\mathbfit{x},\mathbfit{y})$ is a finite but nonempty conjunction of relational $\mathscr{S}$-atoms. An {\em $\mathscr{S}$-UCQ} is a first-order formula built upon $\mathscr{S}$-atoms by using connectives $\wedge,\vee$ and quantifier $\exists$ only. Clearly, every UCQ is equivalent to a disjunction of CQs, and every CQ is also a UCQ. Note that constants are allowed to appear in a query. Given a query (CQ or UCQ) $q$, let $const(q)$ denote the set of all constants that occur in $q$.

A UCQ is called {\em Boolean} if it has no free variables. Let BCQ be short for Boolean CQ. Given a BCQ $q$, let $[q]$ denote a database that consists of all atoms in $q$ where each variable is regarded as a null. In this paper, unless otherwise stated, we only consider Boolean queries. Let $\mathsf{CQ}$ (resp., $\mathsf{UCQ}$) denote the class of Boolean CQs (resp., Boolean UCQs).

\subsubsection{Existential Rule Languages} 


Let $\mathscr{S}$ be a schema. Then every {\em disjunctive embedded dependency (DED)} over $\mathscr{S}$ is a first-order sentence $\sigma$ of the form
\begin{equation}\label{eqn:ded}
\forall\mathbfit{x}\forall\mathbfit{y}(\phi(\mathbfit{x},\mathbfit{y})\rightarrow\exists\mathbfit{z}(\psi_1(\mathbfit{x},\mathbfit{z})\vee\cdots\vee\psi_k(\mathbfit{x},\mathbfit{z}))
\end{equation}
where $\mathbfit{x},\mathbfit{y},\mathbfit{z}$ are tuples of variables, $\phi$ a conjunction of relational $\mathscr{S}$-atoms involving terms only from $\mathbfit{x}\cup\mathbfit{y}$, each $\psi_i$ a conjunction of $\mathscr{S}$-atoms involving terms only from $\mathbfit{x}\cup\mathbfit{z}$, and every variable in $\mathbfit{x}$ has at least one occurrence in $\phi$. For simplicity, we omit universal quantifiers and brackets outside the atoms. Let $head(\sigma)=\{\psi_i:1\le i\le k\}$ and $body(\sigma)=\phi$, called the {\em head} and  {\em body} of $\sigma$, respectively.

{\em Disjunctive tuple-generating dependencies (DTGDs)} are defined as equality-free DEDs, and {\em tuple-generating dependencies (TGDs)} are disjunction-free DTGDs. A TGD is  called {\em linear} if its body consists of a single atom. A DED of the form~(\ref{eqn:ded}) is {\em canonical} if $\psi_i,  1\le i\le k,$ consists of a single atom. It is well-known that, by introducing  auxiliary relation symbols, every set of DEDs (resp, DTGDs, TGDs and linear TGDs) can be converted to an equivalent (under query answering) set of canonical DEDs (resp., DTGDs, TGDs and linear TGDs). Hence, unless stated otherwise, we assume {\em dependencies are canonical} in the rest of this paper. 

Let $D$ be a database, $\Sigma$ a set of DEDs, and ${q}$ a Boolean UCQ. We write $D\cup\Sigma\vDash{q}$ if, for all instances $I$, if $D\subseteq I$ and $I$ is a model of $\sigma$ for all $\sigma\in\Sigma$, then $I$ is also a model of ${q}$, where the notion of {\em model} is defined in a standard way. 

\subsubsection{OMQA-ontologies}

In this subsection, we introduce some notions related to OMQA-ontology. For more details, please refer to~\cite{ZhangZY16}. Let $\mathscr{D}$ and $\mathscr{Q}$ be a disjoint pair of schemas, and $\mathcal{Q}$ a class of Boolean UCQs. Every {\em quasi-OMQA$[\mathcal{Q}]$-ontology} over $(\mathscr{D},\mathscr{Q})$ is a set of database-query pairs $(D,{q})$, where $D$ is a nonempty $\mathscr{D}$-database and ${q}$ a Boolean $\mathscr{Q}$-UCQ in $\mathcal{Q}$ such that $\textit{const}(q)\subseteq{adom}(D)$. Furthermore, an {\em OMQA$[\mathcal{Q}]$-ontology} is a quasi-OMQA$[\mathcal{Q}]$-ontology $O$ that admits the following properties: 
\begin{enumerate}
\item ({\em Closure under Query Conjunctions}) If ${p}\wedge{q}\in\mathcal{Q}$, $(D,{p})\in O$ and $(D,{q})\in O$, then $(D,{p}\wedge{q})\in O$;
\item ({\em Closure under Query Implications}) If ${p}\in\mathcal{Q}$, ${q}\vDash p$ and $(D,{q})\in O$, then
$(D,p)\in O$;
\item ({\em Closure under Injective Database Homomorphisms}) If $(D,{q})\in O$ and $D\rightarrowtail_{\textit{const(q)}} D'$, then
$(D',{q})\in O$;
\item ({\em Closure under Constant Renaming}) If $(D,{q})\in O$ and $\tau$ is a {\em constant renaming} (i.e., a partial injective function from $\Delta$ to $\Delta$), then 
$(\tau(D),\tau(q))\in O$.
\end{enumerate}


Given a set $\Sigma$ of DEDs, let $[\![\Sigma]\!]_{\mathscr{D},\mathscr{Q}}^\mathcal{Q}$ denote the set of all database-query pairs $(D,q)$ where $D$ is a $\mathscr{D}$-database, $q\in\mathcal{Q}$ a $\mathscr{Q}$-UCQ, and $D\cup\Sigma\vDash q$. Given an OMQA$[\mathcal{Q}]$-ontology $O$ over $(\mathscr{D},\mathscr{Q})$, we say $O$ is {\em defined by} $\Sigma$ if $O=[\![\Sigma]\!]_{\mathscr{D},\mathscr{Q}}^\mathcal{Q}$.

The following characterization for DEDs was established in~\cite{ZhangZY16,Zhang2020}.
\begin{thm}[\citeauthor{ZhangZY16}]
An OMQA$[\mathsf{UCQ}]$-ontology
 is defined by a finite set of DEDs iff it is recursively enumerable.
\end{thm}

For convenience, given a class $\mathcal{Q}$ of Boolean UCQs, every {\em DED$[\mathcal{Q}]$-ontology} (resp., {\em DTGD$[\mathcal{Q}]$-ontology} and {\em TGD$[\mathcal{Q}]$-ontology}) is defined as an OMQA$[\mathcal{Q}]$-ontology which is defined by some finite set of DEDs (resp., DTGDs and TGDs).


\section{DTGDs}
In this section, we examine the program expressive power of DTGDs. To do it, we first present a novel chase algorithm for DTGDs, which also plays a key role in the next section.

\subsubsection{Nondeterministic Chase}

Let $\mathscr{D}$ be a schema. A {\em nondeterministic fact} (over $\mathscr{D}$) is a finite disjunction of ($\mathscr{D}$-)facts. For convenience, we often regard each nondeterministic fact as a set of ground atoms. Every {\em nondeterministic instance} (over $\mathscr{D}$) is defined as a set of nondeterministic facts (over $\mathscr{D}$).

Let $I$ be a nondeterministic instance, and $\sigma$ a DTGD in which $\alpha_1,\dots,\alpha_n$ list all the atoms in the body. We say $\sigma$ is {\em applicable to $I$} if there is a substitution $h$ and a tuple $\mathbfit{F}$ of nondeterministic facts $F_1,\dots,F_n\in I$ such that $h(\alpha_i)\in F_i$ for all $i=1,\dots,n$. In this case, we let $res(\mathbfit{F},\sigma,h)$ denote the nondeterministic fact defined as follows:\vspace{-.2cm}
\begin{equation*}
h'(\textit{head}(\sigma))\cup\bigcup_{i=1}^n F_i\setminus\{h(\alpha_i)\}\vspace{-.2cm}
\end{equation*}
where $h'$ is a substitution that extends $h$ by mapping each existential variable $v$ in $\sigma$ to a null which one-one corresponds to the triple $(\sigma,h(\mathbfit{x}),v)$, 
and $\mathbfit{x}$ denotes the tuple of variables occurring in both the head and the body of $\sigma$. In addition, we call $res(\mathbfit{F},\sigma,h)$ a {\em result of applying $\sigma$ to $I$}. 

Furthermore, given a database $D$ and a set $\Sigma$ of DTGDs, let $chase_0(D,\Sigma)=D$; for $k>0$ let $chase_{k}(D,\Sigma)$ denote the union of $chase_{k-1}(D,\Sigma)$ and the set of all results of applying $\sigma$ to $chase_{k-1}(D,\Sigma)$ for all $\sigma\in\Sigma$. Let $chase(D,\Sigma)$ denote the union of $chase_k(D,\Sigma)$ for all $k\ge 0$.

In above definitions, if $\Sigma$ is a set of TGDs, the procedure of nondeterministic chase will degenerate into the traditional oblivious skolem chase, see, e.g.,~\cite{Marnette2009}.

The following theorem gives the soundness and completeness of the nondeterministic chase.
\begin{thm}\label{thm:sound_comp}
Let $\Sigma$ be a set of DTGDs, $D$ be a database, and $q$ be a Boolean UCQ. Then $D\cup\Sigma\vDash q$ iff ${chase}(D,\Sigma)\vDash q$, where by the notation ${chase}(D,\Sigma)\vDash q$ we denote that $q$ is a logical consequence of ${chase}(D,\Sigma)$ as usual.
\end{thm}

Now we generalize the notion of homomorphism from instances to nondeterministic instances. Let $I$ and $J$ be nondeterministic instances over the same schema. Given a set $C$ of constants, a function $h:adom(I)\rightarrow adom(J)$ is called a {\em $C$-homomorphism} from $I$ to $J$, written $h:I\rightarrow_C J$, if we have $h(I)\subseteq J$ and $h(c)=c$ for all constants $c\in C$. 

The following proposition shows that the nondeterministic chase preserves generalized homomorphisms. This property will play an important role in our first characterization. 

\begin{prop}\label{lem:chase_hom_prv}
Let $\Sigma$ be a set of DTGDs, let $D$ and $D'$ be databases, and let $C$ be a set of constants. If there exists a $C$-homomorphism $\tau$ from $D$ to $D'$, then there exists a $C$-homomorphism $\tau'\supseteq\tau$ from  $chase(D,\Sigma)$ to $chase(D',\Sigma)$.
\end{prop}

\subsubsection{Characterization}  In this subsection, we establish a characterization for DTGDs.
%
Before proceeding, we need to present some properties for OMQA-ontologies.

Let $\mathcal{Q}$ be a class of UCQs.  An OMQA$[\mathcal{Q}]$-ontology $O$ is said to be {\em closed under database homomorphisms} if, for all $(D,q)\in O$, if $D'$ is a database with $D\rightarrow_{\textit{const}(q)} D'$, then $(D',q)\in O$; and $O$ is {\em closed under constant substitutions} if, for all $(D,q)\in O$, if $\tau$ is a {\em constant substitution} (i.e., a partial function from $\Delta$ to $\Delta$), then $(\tau(D),\tau(q))\in O$.


The following two propositions tell us that ontologies defined by DTGDs are closed under both of above properties.

\begin{prop}\label{prop:close_dh}
Every DTGD$[\mathsf{UCQ}]$-ontology is closed under database homomorphisms.
\end{prop}

\begin{prop}\label{prop:close_cs}
Every DTGD$[\mathsf{UCQ}]$-ontology is closed under constant substitutions.
\end{prop}

Moreover, we can show that the properties above exactly capture the class of DED-ontologies definable by DTGDs. 

\begin{thm}\label{thm:char_dtgd}
A DED$[\mathsf{UCQ}]$-ontology is defined by a finite set of DTGDs iff it is closed under both database homomorphisms and constant substitutions.\vspace{-.1cm}
\end{thm}

\begin{proof}[Sketch of Proof]
The direction of ``only-if" immediately follows from Propositions~\ref{prop:close_dh} and~\ref{prop:close_cs}. It thus remains to consider the converse. 
Let $O$ be a DED$[\mathsf{UCQ}]$-ontology closed under both database homomorphisms and constant substitutions, and let $\Sigma$ be a finite set of DEDs that defines $O$. We need to construct a finite set $\Sigma'$ of DTGDs which plays the same role as $\Sigma$ under the semantics of UCQ-answering. 

To implement the construction, we introduce $\textit{Eq}$ as a fresh binary relation symbol, and use some constraints to assure that $\textit{Eq}$ defines an equivalence relation. Clearly, such constraints can be represented by several DTGDs in a routine way. Furthermore, for every ($k$-ary) relation symbol $R$ occurring in $\Sigma$, we employ a DTGD of the form\vspace{-.05cm}
\begin{equation}\label{eqn:eq_diffusion}
{\wedge}_{i=1}^k\textit{Eq}(x_i,y_i)\wedge{R}(x_1,\dots,x_k)\rightarrow{R}(y_1,\dots,y_k)\vspace{-.05cm}
\end{equation}
to assure that all terms (constants or nulls) equivalent w.r.t. $\textit{Eq}$ will play the same role in $R$. 

Moreover, we simulate each DED $\sigma\in\Sigma$ by a DTGD $\sigma^\ast$, which is obtained from $\sigma$ by substituting $\textit{Eq}$ for every occurrence of the equality symbol $=$. Let $\Sigma'$ be the set consisting of all the DTGDs mentioned above.
%
%
%
Thanks to the closure of $O$ under both database homomorphisms and constant substitutions, one can prove that the transformation preserves the semantics of UCQ-answering, i.e., $D\cup\Sigma\vDash q$ iff $D\cup\Sigma'\vDash q$ for all $\mathscr{D}$-databases $D$ and Boolean $\mathscr{Q}$-UCQs $q$. Thus, $\Sigma'$ is the desired DTGD set, which completes the proof.
\end{proof}

Let $\mathsf{UCQ}^-$ denote the class of all Boolean UCQs involving no constant. For query answering with queries in $\mathsf{UCQ}^-$, the above characterization can be simplified as follows: 

\begin{cor}\label{cor:char_dtgd}
A DED$[\mathsf{UCQ}^-]$-ontology is defined by a finite set of DTGDs iff it is closed under database homomorphisms.\vspace{-.2cm}
\end{cor}

\section{TGDs}

In this section, let us consider another important existential rule language TGDs, a sublanguage of DTGDs in which disjunctions are not allowed to appear in the rule head. 

\subsubsection{Characterization for CQ-answering}


We first show that, in the case of CQ-answering, disjunctions can be removed from DTGDs. In other words, TGDs have the same expressive power as DTGDs under CQ-answering. 

\begin{thm}\label{thm:dtgd2tgd_cqa}
Every DTGD$[\mathsf{CQ}]$-ontology is defined by a finite set of TGDs.
\end{thm}

To prove this theorem, it suffices to translate every set of DTGDs to a set of TGDs such that they define the same ontology under CQ-answering. Suppose $O$ is a DTGD$[\mathsf{CQ}]$-ontology over a schema pair $(\mathscr{D},\mathscr{Q})$, and $\Sigma$ a set of canonical DTGDs that defines $O$. The general idea is to construct a set ${\Sigma^\ast}$ of TGDs such that the deterministic chase on $\Sigma^\ast$ simulates the nondeterministic chase on $\Sigma$. The desired simulation employs a technique used in Section 3 of~\cite{ZhangZ2017} in which the progression of disjunctive logic programs is simulated by normal logic programs. The main difficulty here is that we need to treat CQ-answering.


To encode a nondeterministic fact, we need a set of numbers and an encoding function. The encoding function is defined by a ternary relation symbol $\textit{Enc}$. By $\textit{Enc}(x,y,z)$ we mean that $z$ encodes the pair $(x,y)$. Numbers used in the encoding are collected by a unary relation symbol $\textit{Num}$. Note that numbers here are not necessary to be natural numbers. For a technical reason, we also use a unary relation symbol $\textit{GT}$ to collect the set of all ground terms that would be used.  
Next, we show how to implement the encoding.

\medskip
For every relation symbol $R\in\mathscr{D}$, we introduce the TGDs\vspace{-.1cm}
\begin{eqnarray}
{R}(x_1,\dots,x_k)&\rightarrow&{\wedge}_{i=1}^k\left(\textit{Num}(x_i)\wedge\textit{GT}(x_i)\right)\\
&\rightarrow&\exists x\, \textit{Flag}_R(x)\wedge\textit{Num}(x)\vspace{-.1cm}
\end{eqnarray}
where $k$ is the arity of $R$, and $\textit{Flag}_R$ a unary relation symbol that defines a flag for the relation $R$. The first TGD asserts that all parameters of $R$ are both numbers and ground terms, and the second one asserts that the flag for $R$ must exist and, in particular, it is also a number.

To define the encoding function, we use the TGD\vspace{-.1cm}
\begin{equation}
\textit{Num}(x)\wedge\textit{Num}(y)\rightarrow\exists z\, \textit{Enc}(x,y,z)\wedge\textit{Num}(z)\vspace{-.1cm}
\end{equation}
which asserts that, for all numbers $x$ and $y$, there is a number $z$ to encode the pair $(x,y)$. With the relations defined above, we are then able to encode (ground) atoms. For example, to encode the atom $\alpha=R(x_1,x_2)$, we use the formula\vspace{-.1cm}
$$
\textit{Flag}_R(y_1)\wedge\textit{Enc}(y_1,x_1,y_2)\wedge\textit{Enc}(y_2,x_2,y_3)\vspace{-.1cm}
$$
which asserts that $y_3$ is a number encoding the atom $\alpha$. Note that $\alpha$ is regarded as the triple $(y_1,x_1,x_2)$ where $y_1$ is the flag of $R$, denoting where the encoding of the first element of the tuple is. In addition, to simplify the notation, given a formula $\varphi(z_0,\mathbfit{z})$, we often use $\varphi(\lceil\alpha\rceil,\mathbfit{z})$ to denote\vspace{-.1cm}
\begin{equation}\label{eqn:enc_formula}
\textit{Flag}_R(y_1)\hspace{-.02cm}\wedge\hspace{-.02cm}\textit{Enc}(y_1,\hspace{-.02cm}x_1,\hspace{-.02cm}y_2)\hspace{-.02cm}\wedge\hspace{-.02cm}\textit{Enc}(y_2,\hspace{-.02cm}x_2,\hspace{-.02cm}y_3)\hspace{-.02cm}\wedge\hspace{-.02cm} \varphi(y_3,\hspace{-.02cm}\mathbfit{z}).
\end{equation}

To encode a disjunction (resp., conjunction) of formulas, we need a flag to denote where the encoding of the first disjunct (resp., conjunct) is. To generate such flags, we use 
\begin{eqnarray}
&\rightarrow&\exists x\, \textit{Flag}_{\textit{d}}(x)\wedge\textit{Num}(x)\\
&\rightarrow&\exists x\, \textit{Flag}_{\textit{c}}(x)\wedge\textit{Num}(x)
\end{eqnarray}
where $\textit{Flag}_{\textit{d}}$ (resp., $\textit{Flag}_{\textit{c}}$) is a unary relation symbol intended to define the flag of encoding disjunction (resp., conjunction). The way of encoding a disjunction (conjunction) is similar to that for atoms, but with a different flag. In addition, the notation $\lceil\cdot\rceil$ can also be extended to disjunctions and conjunctions in an obvious way.  

With the above relations, we are able to encode nondeterministic facts. To access nondeterministic facts, some relations are needed. We introduce fresh relation symbols $
\textit{NF}$, $\textit{Mrg}$ and $\textit{Eq}$. By $\textit{NF}(x)$ we mean that $x$ encodes a nondeterministic fact. By $\textit{Mrg}(x,y,z)$ we denote that $z$ encodes a disjunction (which is a nondeterministic fact) of the nondeterministic facts encoded by $x$ and $y$. Moreover, $\textit{Eq}(x,y)$ asserts that the nondeterministic facts encoded by $x$ and $y$ are equivalent, i.e., they consist of the same set of ground atoms. We only show how to define the merging operation: 
\begin{eqnarray}
\textit{NF}(x)\wedge\textit{Flag}_{\textit{d}}(y)\!\!\!\!\!&\rightarrow\!\!\!\!\!&\textit{Mrg}(x,y,x)\\
\!\!\!\!\!\!\!\!\!\!\!\!\textit{Mrg}(x,u,v)\!\wedge\!\textit{Enc}(u,w,y)\!\wedge\!\textit{Enc}(v,w,z)\!\!\!\!\!&\rightarrow\!\!\!\!\!&\textit{Mrg}(x,y,z)
\end{eqnarray}
To simplify the notation, let $\textit{Mrg}(t_1,\dots,t_k;x_k)$ be short for 
$$
\textit{Flag}_{\textit{d}}(x_0)\wedge\textit{Mrg}(t_1,x_0,x_1)\wedge\cdots\wedge\textit{Mrg}(t_k,x_{k-1},x_k).
$$

Next let us construct TGDs to simulate the nondeterministic chase on $\Sigma$. We introduce $\textit{True}$ as a fresh unary relation symbol, and by $\textit{True}(x)$ we mean that the formula  encoded by $x$ can be inferred from the set of nondeterministic facts generated by the chase. For each canonical DTGD $\sigma\in\Sigma$, if $\alpha_1,\dots,\alpha_k$ list all atoms in the body of $\sigma$, we use the following TGDs to simulate the nondeterministic chase for $\sigma$:
\begin{eqnarray}
\begin{aligned}
{\wedge}_{i=1}^k(\textit{NF}(v_i)\wedge\,&\textit{Enc}(u_i,\lceil \alpha_i\rceil,v_i)\wedge\textit{True}(v_i))\\
\wedge\,\textit{Mrg}(u_1&,\dots,u_k; y)\rightarrow\exists\mathbfit{z}\,\textit{T}_\sigma(\mathbfit{x},y,\mathbfit{z})\wedge\textit{Num}(\mathbfit{z})
\end{aligned}\\
\textit{T}_\sigma(\mathbfit{x},y,\mathbfit{z})\wedge\textit{Mrg}(y,\lceil head(\sigma)\rceil,w)\rightarrow\textit{True}(w)
\end{eqnarray}
where $\mathbfit{x}$ (resp., $\mathbfit{z}$) is the tuple of universal (resp., existential) variables in $\sigma$, $\textit{T}_\sigma$ is a fresh relation symbol of arity $|\mathbfit{x}|+1+|\mathbfit{z}|$, and $\textit{Num}(\mathbfit{z})$ is a conjunction of $\textit{Num}(v)$ for all $v\in\mathbfit{z}$.

To initialize the truth of relations over the data schema $\mathscr{D}$, for each $k$-ary $R\in\mathscr{D}$, we introduce the following TGD:
\begin{equation}
\begin{aligned}
R(x_1,\dots,x_k)\wedge\textit{Flag}_{R}(y_0)\wedge\textit{Enc}(y_0,&\,x_1,y_1)\wedge\cdots\\
\wedge\,\textit{Enc}(y_{k-1},x_k,y_k&)\rightarrow\textit{True}(y_k)
\end{aligned}
\end{equation}


To make sure that the equivalent facts play the same role in the chase procedure, we define the following TGD:
\begin{equation}
\textit{NF}(x)\wedge\textit{NF}(y)\wedge\textit{True}(x)\wedge\textit{Eq}(x,y)\rightarrow\textit{True}(y)
\end{equation}

Let $\Sigma'$ denote the set of all TGDs defined above. Fix a database $D$. By definition, it is easy to see that symbol $\textit{Enc}$ defines an encoding function in $\textit{chase}(D,\Sigma')$. That is, for all numbers $a,b$ defined by $\textit{Num}$ in $\textit{chase}(D,\Sigma')$, there is exactly one term $c$ such that $\textit{Enc}(a,b,c)$ holds in $\textit{chase}(D,\Sigma')$. Moreover, each symbol in $\textit{Flag}_R,\textit{Flag}_{\textit{d}},\textit{Flag}_{\textit{c}}$ defines exactly one number (called a flag) in $\textit{chase}(D,\Sigma')$. Given a nondeterministic fact $F$, let $\langle F\rangle$ denote the number encoding $F$ under the defined encoding function and flags. By an induction on  chase, one can prove the following:
\begin{lem}
$F\!\in\! chase(D,\hspace{-0.04cm}\Sigma)$ iff $\textit{True}(\langle F\rangle)\!\in\! chase(D,\hspace{-0.04cm}\Sigma')$.
\end{lem}

With this lemma, to construct the desired TGD set $\Sigma^\ast$, it remains to define some TGDs which generate the BCQs derivable from $chase(D,\Sigma)$. The following property will play an important role in implementing this task.

\begin{lem}\label{lem:gcq2bcq}
Let $\Sigma$ be a finite set of DTGDs, $D$ a database, and $q$ a BCQ of the form $\exists\mathbfit{x}\varphi(\mathbfit{x})$ where $\varphi$ is quantifier-free and $\mathbfit{x}$ is a tuple of length $k$ which lists all the free variables in $\varphi$. Then $D\cup\Sigma\vDash q$ iff there exists a finite set $T\subseteq term(chase(D,\Sigma))^k$ such that $chase(D,\Sigma)\vDash\bigvee_{\mathbfit{t}\in T}\varphi(\mathbfit{t})$.  
\end{lem}

\comment{
\begin{proof}
By the soundness and completeness of chase, we have $D\cup\Sigma\vDash q$ iff $chase(D,\Sigma)\vDash q$. By compactness, the latter holds iff there exists $n\ge 0$ such that $chase_n(D,\Sigma)\vDash q$. Let $T$ denote $term(chase_n(D,\Sigma))^k$. It is easy to check that $chase_n(D,\Sigma)\vDash q$ iff $chase_n(D,\Sigma)\vDash\bigvee_{\mathbfit{t}\in S}\varphi(\mathbfit{t})$. Since $T$ is always a finite set, the lemma follows.
\end{proof}
}

To implement the above idea, we need more relation symbols, including $\textit{DNF}$ and $\textit{Normalize}$. By $\textit{DNF}(x)$ we denote that the (quantifier-free) formula encoded by $x$ is of disjunctive normal form (DNF), and by $\textit{Normalize}(x,y,z)$ we mean that $z$ encodes a DNF-formula obtained from the conjunction of (DNF-formulas encoded by) $x$ and $y$ by applying the distributive law. Such relations can be defined in TGDs by recursions in a routine way. We omit the details here.


To encode BCQs, we need to generate an infinite number of variables, which can be done by the following TGDs:
\begin{eqnarray}
&\rightarrow&\exists x\, \textit{Var}(x)\wedge\textit{Num}(x)\\
 \textit{Var}(x)&\rightarrow&\exists y\,\textit{Next}(x,y)\wedge\textit{Var}(y)\wedge\textit{Num}(y)
\end{eqnarray}
where $\textit{Var}(x)$ asserts that $x$ is a variable, and $\textit{Next}(x,y)$ denotes that $y$ is the variable immediately after $x$. The generated variables will be used as numbers. Furthermore, we use $\textit{BCQ}(x)$ to denote that $x$ encodes a BCQ. Note that all variables in a BCQ are existential, so we can omit the quantifiers, and simply regard it as a finite conjunction of atoms.

In addition, we introduce a fresh binary relation symbol $\textit{Match}$. By $\textit{Match}(x,y)$ we mean that $y$ encodes a ground DNF-formula in which each disjunct $\psi$ is an instantiation of the BCQ $q$ encoded by $x$, that is, $\psi$ can be obtained from $q$ by substituting some ground term for each existential variable.  

With the above relations, we are now able to generate all the numbers encoding BCQs derivable from $chase(D,\Sigma)$.
\begin{eqnarray}
\textit{True}(x)\wedge\textit{True}(y)\wedge\textit{Normalize}(x,y,z)&\!\!\!\!\!\rightarrow\!\!\!\!\!&\textit{True}(z)\\
\!\!\!\!\!\!\!\!\!\!\textit{BCQ}(x)\wedge\textit{DNF}(y)\wedge\textit{True}(y)\wedge\textit{Match}(x,y)&\!\!\!\!\!\rightarrow\!\!\!\!\!&\textit{True}(x)
\end{eqnarray}

To make sure that the BCQs encoded by this class of numbers are derivable from $chase(D,\Sigma^\ast)$, we employ Zhang {\em et al.}'s technique of generating universal model (see Subsection 5.4 and Proposition 11 in~\cite{ZhangZY16}). Given a class $\mathbb{K}$ of databases over the same schema and a set $C$ of constants, let $\bigoplus_C\mathbb{K}$ denote the {\em $C$-disjoint union} of $\mathbb{K}$, that is, the instance $\bigcup\{D^\ast:D\in\mathbb{K}\}$ where, for every $D\in\mathbb{K}$, $D^\ast$ is an isomorphic copy of $D$ such that, for each pair of distinct databases $D_1$ and $D_2$ in $\mathbb{K}$, only constants from $C$ will be shared by $D^\ast_1$ and $D^\ast_2$. 

Given an OMQA$[\mathsf{CQ}]$-ontology $O$ and a database $D$ over a proper schema, the {\em universal model} of $O$ w.r.t. $D$, denoted $U_O(D)$, is defined as follows:\vspace{-.1cm}
 $$U_O(D)={\bigoplus}_{adom(D)}\{[q]:(D,q)\in O\}.\vspace{-.1cm}$$

\begin{lem}[\citeauthor{ZhangZY16} \citeyear{ZhangZY16}]\label{lem:query2um}
Let $O$ be an OMQA$[\mathsf{CQ}]$-ontology $O$ over a schema pair $(\mathscr{D},\mathscr{Q})$, $D$ a $\mathscr{D}$-database and $q$ a $\mathscr{Q}$-BCQ. Then $(D,q)\in O$ iff $U_O(D)\models q$.
\end{lem}

\comment{

\begin{proof}
The direction of ``only-if" is trivial; so we only show the converse. Let ${q}$ be a $\mathscr{Q}$-BCQ and $D$ a $\mathscr{D}$-database such that $U_O(D)\models{q}$. By definition, there exists a homomorphism $h$ from $[{q}]$ to $U_O(D)$ such that $h(c)=c$ for all constants $c$ from $C$. Let ${q}'$ be a BCQ obtained from ${q}$ by substituting $h(x)$ for all variables $x$ such that $h(x)\in C$, and let $h'$ be a restriction of $h$ to all variables and constants occurring in ${q}'$. Then it is clear that ${q}'\models{q}$ and that $h'$ is a homomorphism from $[{q}']$ to $U_O(D)$. W.l.o.g., suppose ${q}'$ is of the form ${q}'_1\wedge\cdots\wedge {q}'_k$, 
where each ${q}'_i$ is prime. 
Take $i\in\{1,\dots,k\}$. Obviously, $h'$ is a homomorphism from $[{q}'_i]$ to $U_O(D)$. Since ${q}'_i$ is prime, by the definition of $U_O(D)$ we know that there is at least one BCQ ${p}_i$ such that $(D,{p}_i)\in O$ and that $h'({q}'_i)$ is contained in the disjoint copy of ${p}_i$ in $U_O(D)$. From the latter we conclude that ${p}_i\models{q}'_i$. Since $O$ is closed under query implications, we infer that $(D,{q}'_i)\in O$. Furthermore, since $O$ is closed under query conjunctions, we then obtain that $(D,{q}')\in O$. As we have proved previously, it holds that ${q}'\models{q}$. Again by the query implication closure of $O$, we conclude that $(D,{q})\in O$, which completes the proof. 
\end{proof}

}

With the above lemma, it remains to show how to generate the universal model $U_O(D)$
.
Let $a$ be a number that encodes a BCQ ${q}$ such that $\textit{True}(a)$ holds in the intended instance. For all $\mathscr{Q}$-atoms $\alpha$, we first test whether $\alpha$ appears in ${q}$. If the answer is yes we then copy $\alpha$ to the universal model. Since $U_O(D)$ is defined by a disjoint union of $[{q}]$, a renaming of variables in ${q}$ would be necessary, which can be achieved by using existential variable in the rule head to generate nulls. We introduce a relation symbol $\textit{Ren}$, and by $\textit{Ren}(y,z,x)$ we mean that $y$ will be replaced with $z$ in the copy of BCQ (encoded by) $x$. Below are some TGDs to implement it:
\begin{eqnarray}
\textit{BCQ}(x)\wedge\textit{Var}(y)\rightarrow\exists z\,\textit{Ren}(y,z,x)\\
\textit{BCQ}(x)\wedge\textit{GT}(y)\rightarrow\textit{Ren}(y,y,x)
\end{eqnarray}
where the second TGD means that all the constants appearing in the BCQ will not be changed in the copy.

To generate the universal model $U_O(D)$, we still need to introduce a relation symbol $\textit{Has}Q$ for each relation symbol $Q\in\mathscr{Q}$. By $\textit{Has}Q(\mathbfit{y},x)$ we mean that $Q(\mathbfit{y})$ is an atom appearing in the BCQ encoded by $x$. By traversing the whole BCQ, it is easy to see that $\textit{Has}Q$ can be defined by TGDs. To copy all the atoms involving $Q$ and appearing in the BCQ to the universal model, we employ the following TGD:
\begin{eqnarray}
\textit{BCQ}(x)\hspace{-0.03cm}\wedge\hspace{-0.03cm}\textit{True}(x)\hspace{-0.03cm}\wedge\hspace{-0.03cm}\textit{Has}Q(\mathbfit{y},x)\hspace{-0.03cm}\wedge\hspace{-0.03cm}\textit{Ren}(\mathbfit{y},\!\mathbfit{z},\!x)\rightarrow{Q}(\mathbfit{z})
\end{eqnarray}
where $\textit{Ren}(\mathbfit{y},\mathbfit{z},x)$ denotes formula $\bigwedge_{1\le j\le k}\textit{Ren}(y_j,z_j,x)$ if $\mathbfit{y}=y_1\cdots y_k$, $\mathbfit{z}=z_1\cdots z_k$, and $k$ is the arity of ${Q}$.

Let $\Sigma^\ast$ be the set of TGDs defined in this subsection. Then the following property holds, which yields Theorem~\ref{thm:dtgd2tgd_cqa}.

\begin{prop}
For every pair of $\mathscr{D}$-database $D$ and $\mathscr{Q}$-BCQ $q$, we have $chase(D,\Sigma)\vDash q$ iff $chase(D,\Sigma^\ast)\vDash q$.
\end{prop}

\subsubsection{Characterization for UCQ-answering}

It is worth noting that the translation proposed in the last subsection does not work for UCQ-answering. In this subsection, we examine the expressive power of TGDs for this case. 

We first define a property. An OMQA$[\mathsf{UCQ}]$-ontology $O$ is said to {\em admit query constructivity} if $(D,p\vee q)\in O$ implies either $(D,p)\in O$ or $(D,q)\in O$.
The following theorem tells us that the above property exactly captures the definability of a DTGD$[\mathsf{UCQ}]$-ontology by TGDs.

\begin{thm}\label{thm:char_tgds}
A DTGD$[\mathsf{UCQ}]$-ontology is defined by a finite set of TGDs iff it admits query constructivity.
\end{thm}

To prove this theorem, we need some notation and property. Given an OMQA$[\mathsf{UCQ}]$-ontology $O$, let $O|_{\mathsf{CQ}}$ denote 
$\{(D,q)\in O:q\in\mathsf{CQ}\}$ which is an OMQA$[\mathsf{CQ}]$-ontology.

\begin{lem}\label{lem:cq2ucq}
Let $O$ and $O'$ be OMQA$[\mathsf{UCQ}]$-ontologies that admit query constructivity. If $O|_{\mathsf{CQ}}\!=\!O'|_{\mathsf{CQ}}$ then $O\!=\!O'$.
\end{lem}

\comment{
\begin{proof}
Suppose $O|_{\mathsf{CQ}}=O'|_{\mathsf{CQ}}$. We need to prove $O=O'$. Let $(D,q)\in O$. Due to the symmetry, it suffices to show that $(D,q)\in O'$. Suppose $q$ is of the form $q_1\vee\cdots\vee q_n$ where each $q_i$ is a BCQ. By the query constructivity of $O$, there must exist $i\in\{1,\dots,n\}$ such that $(O,q_i)\in O$, which implies $(O,q_i)\in O|_{\mathsf{CQ}}$ immediately. By assumption, we then obtain $(O,q_i)\in O'|_{\mathsf{CQ}}$. Consequently, we have $(O,q_i)\in O'$. On the other hand, it is clearly true that $q_i\vDash q$. Since $O'$ is closed under query implications, we thus conclude that $(O,q)\in O'$, which completes the proof.
\end{proof}
}

Now we are in the position to prove Theorem~\ref{thm:char_tgds}.

\begin{proof}[Proof of Theorem~\ref{thm:char_tgds}]
The direction of ``if" follows from Lemma~\ref{lem:cq2ucq} and Theorem~\ref{thm:dtgd2tgd_cqa}. For the converse, we assume $O$ is defined by a finite set $\Sigma$ of TGDs. Let $(D,p\vee q)\in O$, where $p$ and $q$ are Boolean UCQs. By the completeness of the chase, it holds that $chase(D,\Sigma)\models p\vee q$. Note that $chase(D,\Sigma)$ here is a deterministic instance. We thus have either $chase(D,\Sigma)\models p$ or $chase(D,\Sigma)\models q$. By the soundness of the chase, either $(D,p)\in O$ or $(D,q)\in O$ must be true, which yields the desired direction .
\end{proof}

\begin{example}
Let $\mathscr{D}$ be the schema $\{P\}$, and $\mathscr{Q}$ be the schema $\{Q,R\}$, where $P,Q$ and $R$ are unary relation symbols. Let $\Sigma$ be a set consisting of a single DTGD defined as follows:
\begin{equation}
P(x)\rightarrow Q(x)\vee R(x)
\end{equation}
Let $D=\{P(a)\}$. Clearly, $D\cup\Sigma\vDash Q(a)\vee R(a)$, but neither $D\cup\Sigma\vDash Q(a)$ nor $D\cup\Sigma\vDash R(a)$. So the ontology defined by $\Sigma$ over $(\mathscr{D},\mathscr{Q})$ does not admit query constructivity.
\end{example}

By the above example and Theorem~\ref{thm:char_tgds}, we thus have:

\begin{cor}
There is a DTGD$[\mathsf{UCQ}]$-ontology that is not defined by any finite set of TGDs.
\end{cor}

The next corollary immediately follows from Theorem~\ref{thm:char_tgds} and Lemma~\ref{lem:cq2ucq}. With it, to examine the expressive power of TGDs, we need only to consider CQ-answering. In the next section, we will thus focus on CQ-answering.
\begin{cor}
Let $\mathscr{D}$ and $\mathscr{Q}$ be a pair of schemas. Let $\Sigma$ and $\Sigma'$ be finite sets of TGDs. Then 
\begin{equation*}
[\![\Sigma]\!]^{\mathsf{UCQ}}_{\mathscr{D},\mathscr{Q}}=[\![\Sigma']\!]^{\mathsf{UCQ}}_{\mathscr{D},\mathscr{Q}}\text{ iff }[\![\Sigma]\!]^{\mathsf{CQ}}_{\mathscr{D},\mathscr{Q}}=[\![\Sigma']\!]^{\mathsf{CQ}}_{\mathscr{D},\mathscr{Q}}.
\end{equation*}
\end{cor}

\section{Linear TGDs}

In this section, we focus on the program expressive power of linear TGDs. 
Before establishing the characterization, we need to recall some notions and make a few assumptions.

\subsubsection{Tree Automata} 
First recall some notions of tree automata. For more details, please refer to, e.g.,~\cite{tata2007}. 

Let $\mathcal{L}$ be a nonempty set of labels. An {\em $\mathcal{L}$-labeled tree} $T$ is a quadruple $(V,E,r,L)$ where $E\subseteq V\times V$, $(V,E)$ defines a tree with the {\em root} $r\in V$ in a standard way, and $L:V\rightarrow\mathcal{L}$ is called the {\em label function}. $T$ is called {\em finite} if $V$ is finite.

Every {\em ranked input alphabet} is a finite and nonempty set of {\em input symbols}, each is a pair $\omega=(\ell(\omega),ar(\omega))$, where $\ell(\omega)$ is the {\em letter} of $\omega$, and $ar(\omega)$ a natural number called the {\em arity} of $\omega$. Given a ranked input alphabet $\Omega$, an {\em $\Omega$-ranked tree} is a finite labeled tree $\mathfrak{T}=(V,E,r,L)$ over $\Omega$ such that every node $v\in V$ has exactly $ar(L(v))$ children in $\mathfrak{T}$. 

For convenience, we often use expressions built over $\Omega$ to denote ranked trees. A nullary input symbol $\pi\in\Omega$ denotes a ranked tree consisting of a single node with the label $\pi$. Let $\omega\in\Omega$ be a $k$-ary symbol, $e_1,\dots,e_k$ be expressions denoting $\Omega$-ranked trees $\mathfrak{T}_1,\dots,\mathfrak{T}_k$. We then use the expression $\omega(e_1,\dots,e_k)$ to denote the $\Omega$-ranked tree $\mathfrak{T}$, in which the root $r$ is labeled as $\omega$, such that for every $i=1,\dots,k$, $\mathfrak{T}_i$ is a subtree of $\mathfrak{T}$ and the $i$-th child of $r$ is the root of $\mathfrak{T}_i$.

Moreover, a {\em nondeterministic (bottom-up) tree automaton} (NTA) $\mathcal{A}$ is defined as a quadruple $(S, F, \Omega, \Theta)$ where
\begin{enumerate} 
\item $S$ is a finite set of {\em states};
\item $F\subseteq S$ is a set of {\em final state}s;
\item $\Omega$ is a ranked input alphabet; 
\item $\Theta\subseteq \Omega\times S^\ast\times S$ is a transition relation which consists of
{\em transition rules} of the form $
(\omega, (s_1,\dots,s_k), s_0)
$, 
where 
$\omega\in\Omega$ is a $k$-ary symbol for some $k$ and $s_0,\dots,s_k\in S$.
\end{enumerate}
%

Let $e$ and $e'$ be expressions built over $\Omega$ and $S$, where states in $S$ are regarded as unary symbols. We say $e'$ is a {\em legal transition} from $e$ if there is an $\Omega$-ranked tree $t$ and a transition rule $(\omega, \mathbfit{s}, s')\in\Theta$ such that $e\ne e'$ and $e'$ is obtained from $e$ by substituting $s'(\omega(\mathbfit{t}))$ for exactly one occurrence of $\omega(s_1(t_1),\dots,s_k(t_k))$, where both $\mathbfit{s}$ and $\mathbfit{t}$ are $k$-tuples for some $k$, and $s_i$ (resp., $t_i$) is the $i$-th component of $\mathbfit{s}$ (resp., $\mathbfit{t}$). 
Every {\em run} of $\mathcal{A}$ on an $\Omega$-ranked tree $t$ is a finite sequence of expressions $e_0,\dots,e_n$ such that $e_0=t$, $e_i$ is a legal transition from $e_{i-1}$ for $0<i\le n$, and there is no legal transition from $e_n$. 

An NTA $\mathcal{A}=(S,F,\Omega,\Theta)$ is said to {\em accept} an $\Omega$-ranked tree $\mathfrak{T}$ if there is a run $e_0,\dots,e_n$ of $\mathcal{A}$ on $\mathfrak{T}$ and a final state $s\in F$ such that $e_n=s(\mathfrak{T})$. An $\Omega$-ranked tree language $\mathbb{L}$, i.e., a set of $\Omega$-ranked trees, is said to be {\em recognized} by $\mathcal{A}$ if every $\Omega$-ranked tree is accepted by $\mathcal{A}$ if, and only if, it is in $\mathbb{L}$. It is well-known that a ranked tree language is recognized by some NTA iff it is regular, see, e.g.,~\cite{tata2007}.

An NTA $\mathcal{A}$ is called {\em oblivious} if for every pair of transition rules $(\omega,\mathbfit{s},s_0)$ and $(\omega',\mathbfit{s}',s'_0)$ of $\mathcal{A}$, if $\ell(\omega)=\ell(\omega')$ then we have $s_0=s'_0$. In other words, the transition of $\mathcal{A}$ only depends on the letter of the current input symbol. 
Given a ranked tree $\mathfrak{T}=(V,E,r,L)$, the {\em accompanying tree} of $\mathfrak{T}$, denoted $\ell(\mathfrak{T})$, is defined as the labeled tree $(V,E,r,\ell(L))$ where $\ell(L)(v)=\ell(L(v))$ for all $v\in V$. Given a ranked tree language $\mathbb{L}$, the {\em accompanying tree language} of $\mathbb{L}$ is the class of $\ell(\mathfrak{T})$ for all $\mathfrak{T}\in \mathbb{L}$. Interestingly, a ranked tree language is recognized by an oblivious NTA iff it is regular and its accompanying tree language is closed under prefixes.

\subsubsection{Automata That Accept BCQs} 

Let $q$ be a BCQ. Let $\mathcal{L}_q$ denote the set of order pairs $\langle X,\Phi\rangle$ where $X$ is a finite set of variables or constants, and $\Phi\subseteq[q]$.  A {\em tree representation} of $q$ is a finite $\mathcal{L}_q$-labeled tree $\mathfrak{R}=(V,E,r,L)$ such that   
\begin{enumerate}
\item $[q]=\bigcup_{v\in V}L^2(v)$, and $term(L^2(v))\subseteq L^1(v)$ for every $v\in V$, where, for $i\in\{1,2\}$,  by $L^i(v)$ we denote the $i$-th component of $L(v)$;
\item the subgraph of $\mathfrak{R}$ induced by the set $\{v\!\in\! V:t\!\in\! L^1(v)\}$ is connected for every $t\in \Delta\cup\Delta_{\mathrm{v}}$;
\item for all $v\in V$, all constants in $L^1(v)$ also occur in $L^1(r)$.
\end{enumerate} 
%
%
The {\em width} of $\mathfrak{R}$ is the maximum cardinality of $L^1(v)$ for all $v\in V$.
In particular, a tree representation $\mathfrak{R}=(V,E,r,L)$ of $q$ is called {\em linear} if, for each $v\in V$, we have $|L^2(v)|\le 1$. 

Note that a tree representation of $q$ is not necessary a tree decomposition, but based on any tree decomposition of $q$, one can easily construct a tree representation.

Next we show how to encode BCQs as inputs of an NTA. Let $\mathscr{Q}$ be a schema and $q$ a $\mathscr{Q}$-BCQ. Let $\mathfrak{R}=(V,E,r,L)$ be a tree representation of $q$. A rough idea of encoding $q$ is by directly regarding $\mathfrak{R}$ as the accompanied tree of a ranked tree. However, this is infeasible because
the ranked input alphabet is required to be finite, while the BCQs that we have to consider may involve an unbounded number of terms.

A natural idea to resolve the mentioned issue is by reusing variables. For example, suppose $v_1,v_2$ and $v_3$ are nodes in $\mathfrak{R}$ such that $v_2$ is a child of $v_1$, and $v_3$ a child of $v_2$. Suppose 
\begin{align*}
L(v_1)&=(\{x_1,x_2,x_3\},\{R(x_1,x_2,x_3)\}), \\
L(v_2)&=(\{x_2,x_3,x_4\},\{S(x_3,x_4)\}), \\
L(v_3)&=(\{x_3,x_4,x_5\},\{T(x_5,x_4,x_5))\}. 
\end{align*}
By the definition of tree representation, $x_1$ is not allowed to appear in $v_3$ and its descendants. We thus can reuse $x_1$ in $v_3$, and let $L(v_3)=(\{x_3,x_4,x_1\}, \{T(x_1,x_4,x_1)\})$. We assume all the variables occurring in $v_3$ but not in $v_2$ are fresh variables. Clearly, by reusing variables, only $2k$ variables are needed to encode a tree representation of the width $k$. 

Let $\mathcal{V}$ be a set that consists of $2k$ variables. Let $At$ denote the set of $\mathscr{Q}$-atoms involving terms only from $const(q)$ and $\mathcal{V}$. Let $\mathcal{L}$ be a label set consisting of all the pairs $\omega=(X,\Phi)$ such that $X\subseteq const(q)\cup\mathcal{V}$ and $\Phi$ is either $\emptyset$ or $\{\alpha\}$ for some $\alpha\in At$. Clearly, $\mathcal{L}$ is finite. By the technique of reusing variables, $\mathfrak{R}$ can be represented as an $\mathcal{L}$-labeled tree. Suppose $\mathfrak{R}'=(V,E,r,L')$ is the mentioned tree. Let $\mathfrak{T}$ denote the ranked tree $(V,E,r,L^\ast)$ where $L^\ast(v)=(L'(v),n)$ if $v\in V$ has exactly $n$ children. Clearly, from $\mathfrak{T}$ one can easily obtain $q$. We call  $\mathfrak{T}$ a {\em ranked tree representation} of $q$.

We say an NTA $\mathcal{A}$ {\em accepts} $q$ if  $\mathcal{A}$ accepts $\mathfrak{T}$ for some ranked tree representation $\mathfrak{T}$ of $q$; and $\mathcal{A}$ {\em recognizes} a class $\mathcal{C}$ of $\mathscr{Q}$-BCQs if for all $\mathscr{Q}$-BCQs $q$, $\mathcal{A}$ accepts $q$ iff $q\in\mathcal{C}$.

\subsubsection{Characterization}


We first define some notions and notations. A BCQ $q$ is called {\em nontrivial} if $[q]\ne\emptyset$, and $q$ is called a {\em proper subquery} of another BCQ $p$ if $[q]\subsetneq [p]$. A BCQ $q$ is called {\em inseparable} if there are no nontrivial proper subqueries $q_1$ and $q_2$ of $q$ such that $q$ is equivalent to $q_1\wedge q_2$. 
Let $\mathcal{C}$ be a class of BCQs. A BCQ $q\in\mathcal{C}$ is said to be {\em most specific} w.r.t. $\mathcal{C}$ if the following holds:
\begin{itemize}
\item if there is a partial function $s\!:\!\Delta_{\mathrm{v}}\!\rightarrow\!\Delta$ that maps at least one variable occurring in $q$ to a constant, then $s(q)\not\in\mathcal{C}$.
\end{itemize}
In addition, a BCQ $q\in\mathcal{C}$ is said to be {\em prime} w.r.t. $\mathcal{C}$ if it is inseparable and most specific w.r.t. $\mathcal{C}$.



Given an OMQA$[\mathsf{CQ}]$-ontology $O$ and a database $D$, let $O(D)$ denote the class of BCQs $q$ such that $(D,q)\in O$.


%
%

%
%
%
Now we have a characterizations for linear TGDs.

\begin{thm}\label{thm:char_ltgd}
Let $\mathscr{D}$ and $\mathscr{Q}$ be schemas. An OMQA$[\mathsf{CQ}]$-ontology $O$ over $(\mathscr{D},\mathscr{Q})$ is defined by a finite set of linear TGDs iff it admits both of the following properties:\vspace{-.1cm}
\begin{enumerate}
\item (Data Constructivity) If $D$ and $D'$ are $\mathscr{D}$-databases and $q\in O(D\cup D')$ is prime w.r.t. $O(D\cup D')$, then we have either $q\in O(D)$ or $q\in O(D')$.
\item (NTA-recognizability of Queries) For every $\mathscr{D}$-database $D$ with a single fact, there exists an oblivious NTA which recognizes $O(D)$.\vspace{-.15cm}
\end{enumerate}
\end{thm}

\begin{proof}[Sketch of Proof]
Due to space limit, we only give a proof for the direction of ``only-if". Suppose $O$ is defined by a finite set $\Sigma$ of linear TGDs. We need to show that $O$ admits Properties 1 and 2. Property 1 can be proven by a careful induction on the chase. Below we prove that $O$ admits Property 2. 


Let $D$ be a $\mathscr{D}$-database with a single fact. Now let us construct an oblivious NTA that recognizes $O(D)$. Let $\mathscr{S}$ denote the schema of $\Sigma$. Let $k$ be the maximum arity of relation symbols in $\mathscr{S}$. Let $\mathcal{V}$ be the set that consists of pairwise distinct variables $x_1,\dots,x_{2k}$. Let $At$ be the set of all atoms built upon relation symbols from $\mathscr{S}$ and terms from $adom(D)\cup\mathcal{V}$. We introduce $\lceil\log_2 (|At|+2)\rceil$ fresh variables, and let $\mathcal{V}_0$ denote the set that consists of these variables. Let $\iota$ be an injective function from $At$ to $2^{\mathcal{V}_0}\setminus\{\emptyset,\mathcal{V}_0\}$. Thus, every atom in $At$ can be encoded by a set of variables in $\mathcal{V}_0$. 

With the above assumptions, we are now able to define the NTA. Let $S=At\cup\{\diamond\}$ be the set of states, and let $F=\{\diamond\}$ be the set of final state where $\diamond$ is used as the unique final state. Furthermore, let $\mathcal{L}$ be a label set which consists of \vspace{-.15cm}
\begin{enumerate}
\item $(term(\alpha),\{\alpha\})$ for each $\alpha\in At$ which is a $\mathscr{Q}$-atom;
\item $(term(\alpha)\cup\iota(\alpha),\emptyset)$ for each $\alpha\in At$;
\item $(adom(D)\cup\mathcal{V}_0,\emptyset)$. \vspace{-.15cm}
\end{enumerate}
For convenience, let $\lambda:\mathcal{L}\rightarrow S$ be a function that maps each label $\ell\in\mathcal{L}$ of the form 1 or 2 to the atom (state) $\alpha$, and maps the label of the form 3 to the final state $\diamond$. Clearly, $\lambda$ is well-defined. Let $\Omega$ be a ranked input alphabet which consists of ordered pairs $(\ell,m)$ for all $\ell\in\mathcal{L}$ and all $0\le m\le |At|$, where each $(\ell, m)$ is used as an $m$-ary input symbol.

%

Furthermore, let $\Theta$ be a set consisting of\vspace{-.15cm}
\begin{enumerate}
\item $((\ell,1), \alpha, \diamond)$ if $\lambda(\ell)=\diamond$ and $D=\{\alpha\}$;
\item $((\ell,m), (\alpha_1,\dots,\alpha_m), \alpha)$ if $\lambda(\ell)=\alpha$, $0\le m\le |At|$, $\alpha,\alpha_1,\dots,\alpha_m\in At$ and for $1\le i\le m$, $\{\alpha\}\cup\Sigma\vDash\exists\mathbfit{x}_i\alpha_i$, where $\mathbfit{x}_i$ denotes a tuple consisting of all the variables that occur in $\alpha_i$ but not in $\alpha$.\vspace{-.15cm}
\end{enumerate} 

Let $\mathcal{A}=(S,F,\Omega,\Theta)$. Since $\lambda$ is a well-defined function, we know that $\mathcal{A}$ is an oblivious NTA. Next we show that $\mathcal{A}$ recognizes the class $O(D)$. By the definition of $\mathcal{A}$, it is easy to see that every $\mathscr{Q}$-BCQ accepted by $\mathcal{A}$ belongs to $O(D)$. 

Conversely, let $q\in O(D)$. We need to prove that $\mathcal{A}$ accepts $q$. 
Let $\mathfrak{D}$ be a labeled tree constructed as follows:\vspace{-.1cm}
\begin{enumerate}
\item Create the root $r$ with the label $L(r)=(adom(D),D)$;
\item 
For each node $v$ already in $\mathfrak{D}$, if there is an atom $\alpha\in At$ such that $L^2(v)\cup\Sigma\vDash\exists\mathbfit{x}\,\alpha$, 
then create a child $v'$ for $v$ and let $L(v')\!=\!(term(\alpha^\ast),\{\alpha^\ast\})$, where $\alpha^\ast$ is obtained from $\alpha$ by substituting fresh variables for variables in $\mathbfit{x}$.\vspace{-.1cm}
\end{enumerate}
Let $atom(\mathfrak{D})$ be the set of all atoms appearing in $\mathfrak{D}$. By definition we know that $chase(D,\Sigma)$ is $adom(D)$-isomorphic to a subset of $atom(\mathfrak{D})$. 
%
%
Let $C$ denote $const(q)$. As $q\in O(D)$, according to the construction of $\mathfrak{D}$, it is not difficult to prove that $[q]$ is $C$-isomorphic to a subset of $atom(\mathfrak{D})$.

Let $Q$ be a subset of $atom(\mathfrak{D})$ that is $C$-isomorphic to $[q]$. Let $\mathfrak{D}_q$ be a minimal connected subgraph of $\mathfrak{D}$ that covers $Q$ and the root $r$. Suppose $\mathfrak{D}_q=(V,E,r,L)$. Next, let $\mathfrak{R}_q$ be the labeled tree $(V,E,r,L_0)$ where $L_0$ is defined as follows: \vspace{-.6cm}
\begin{enumerate}
\item for the root $r$, let $L_0(r)=(adom(D)\cup\mathcal{V}_0,\emptyset)$;
\item for every $v\in V$ with the label $L(v)=(term(\alpha),\{\alpha\})$, let $L_0(v)=(term(\alpha)\cup\iota(\alpha),\emptyset)$ if $\alpha\not\in Q$, and $L_0(v)=L(v)$ otherwise. \vspace{-.2cm}
\end{enumerate}
Clearly, $\mathfrak{R}_q$ is a finite and linear tree representation of $q$. By the technique mentioned in the last subsection, such a tree can be naturally encoded by an $\Omega$-ranked tree, which can be easily showed to be accepted by $\mathcal{A}$ by a routine check.
\end{proof}

\section{Conclusion and Related Work}

We have established a number of novel characterizations for the program expressive power of DTGDs, TGDs as well as linear TGDs. These results make significant contributions towards a complete picture for the (absolute) program expressive power of existential rule languages. As a byproduct, we have proposed a new chase algorithm called nondeterministic chase for DTGDs, and proved that it is sound and complete for UCQ-answering. Moreover, we have observed that queries derivable from linear TGDs are recognizable by a natural class of tree automata, and this may shed light on optimizing ontology by automata techniques.

Besides the data and program expressive power, there has been some earlier research motivated to characterize other kinds of expressive power of existential rule languages. For example, \citeauthor{CateK2010}~\shortcite{CateK2010} characterized the source-to-target TGDs (a class of acyclic TGDs) and its subclasses under the semantics of schema mapping; by regarding ontology languages as logical languages, \cite{MakowskyV1986,ZhangZJ2020,ConsoleKP2021} established a number of model-theoretic characterizations for existential rule languages, including DEDs, DTGDs, TGDs, equality-generating dependencies, full TGDs, guarded TGDs as well as linear TGDs. 


\section*{Acknowledgements}

We would like to thank anonymous referees for their helpful comments and suggestions. This work was supported by the National Key R\&D Program of China (2020AAA0108504, 2021YFB0300104) and the National Natural Science Foundation of China (61806102, 61972455). 

\bibliography{aaai22}


\appcomment{
\newpage
\onecolumn
\appendix
\section{Appendix: Detailed Proofs}

\subsection{Proof of Theorem~\ref{thm:sound_comp}}

{\noindent\bf Theorem~\ref{thm:sound_comp}.} 
Let $\Sigma$ be a set of DTGDs, $D$ be a database, and $q$ be a Boolean UCQ. Then $D\cup\Sigma\vDash q$ iff ${chase}(D,\Sigma)\vDash q$.
\medskip

To prove this theorem, it suffices to prove the following two lemmas:

\begin{lem}
If $chase(D,\Sigma)\vDash q$ then $D\cup\Sigma\vDash q$.
\end{lem}

\begin{proof}
Suppose $chase(D,\Sigma)\vDash q$. We need to prove $D\cup\Sigma\vDash q$. Let $I$ be a model of $D\cup\Sigma$, and let $C$ be the set of constants occurring in $q$. To yield the desired conclusion, as $q$ is preserved under $C$-homomorphisms, it suffices to show that there is a model $J$ of $chase(D,\Sigma)$ such that $J\rightarrow_C I$. 

Let $J_0=D$, and $\tau_0$ be the identity function with domain $adom(D)$. Clearly, $J_0$ is a minimal model of $chase_0(D,\Sigma)$, and $\tau_0$ is a $C$-homomorphism from $J_0$ to $I$.

For the case $k>0$, let us assume that $J_{k-1}$ is a minimal model of $chase_{k-1}(D,\Sigma)$, and $\tau_{k-1}$ is a $C$-homomorphism from $J_{k-1}$ to $I$. Let $J_k$ denote the smallest superset of $J_{k-1}$ and $\tau_k$ the smallest extension of $\tau_{k-1}$ such that, if $\sigma\in\Sigma$, $body(\sigma)=\{\alpha_1,\dots,\alpha_n\}$ and $h(body(\sigma))\subseteq J_{k-1}$ for some substitution $h$, then both of the following properties hold:
\begin{enumerate}
\item $res(\mathbfit{F},\sigma,h)\in J_k$, where $\mathbfit{F}$ is a tuple of nondeterministic facts $F_1,\dots,F_n$ such that $h(\alpha_i)\in F_i$ for $1\le i\le n$. Note that $J_{k-1}$ is a minimal model of $chase_{k-1}(D,\Sigma)$, which implies the existence of $\mathbfit{F}$ and the applicability of $\sigma$ to $chase_{k-1}(D,\Sigma)$. Let $h'$ be a substitution that extends $h$ in the way defined by $chase$. 
\item $\tau_{k}(h'(x))=g'(x)$ for every existential variable $x$ in $\sigma$, where $g'$ is a substitution extending $\tau_{k-1}\circ h$ such that $g'(head(\sigma))\subseteq I$. Note that $\tau_{k-1}$ is a $C$-homomorphism from $J_{k-1}$ to $I$. We thus have $\tau_{k-1}(h(body(\sigma)))\subseteq I$. As $I$ is a model of $\sigma$, the substitution $g'$ always exists.
\end{enumerate}

By definition, it is easy to see that $J_k$ is a minimal model of $chase_{k}(D,\Sigma)$, and $\tau_k$ is a $C$-homomorphism from $J_k$ to $I$. Let $J=\bigcup_{k\ge 0}J_k$ and $\tau=\bigcup_{k\ge 0}\tau_k$. It is thus not difficult to verify that $J$ is a model of $chase(D,\Sigma)$, and $\tau$ is a $C$-homomorphism from $J$ to $I$. These complete the proof.
\end{proof}

\begin{lem}
If $D\cup\Sigma\vDash q$ then $chase(D,\Sigma)\vDash q$.
\end{lem}

\begin{proof}
Suppose $D\cup\Sigma\vDash q$. To prove $chase(D,\Sigma)\vDash q$, by the preservation of $q$ under $const(q)$-homomorphisms, it suffices to prove that every minimal model of $chase(D,\Sigma)$ is a model of $q$. Let $I$ be a minimal model of $chase(D,\Sigma)$. It thus remains to show that $I$ is a model of $q$.

On the other hand, it is clear that $D\subseteq chase(D,\Sigma)$. By assumption, we have  $chase(D,\Sigma)\cup\Sigma\vDash q$. So, to yield that $I$ is a model of $q$, it suffices to show that $I$ is a model of $\Sigma$. Let $\sigma\in\Sigma$. Since $\Sigma$ is canonical, we assume $\sigma$ is  of the form
\begin{equation*}
\bigwedge_{i=1}^m\vartheta_i(\mathbfit{x},\mathbfit{y})\rightarrow\exists\mathbfit{z}\bigvee_{j=1}^n\alpha_j(\mathbfit{x},\mathbfit{z})
\end{equation*}
where $\vartheta_i,1\le i\le m$, and $\alpha_j, 1\le j\le n$, are atoms. Suppose $h$ is a substitution such that $h(\vartheta_i)\in I$ for all $1\le i\le m$. Now our task is to show that there is some $1\le j\le n$ and a substitution $h'$ extending $h$ such that $h'(\alpha_j)\in I$. 

Before completing the proof, we claim that for each $1\le i\le m$ there exists a nondeterministic fact $F\in chase(D,\Sigma)$ such that $F\cap I=\{h(\vartheta_i)\}$. Otherwise, let $i$ be an index such that $h(\vartheta_i)\not\in F$ for any $F\in chase(D,\Sigma)$; then it is not difficult to verify that $I\setminus\{h(\vartheta_i)\}$ is also a model of $chase(D,\Sigma)$, which contradicts with the minimality of $I$. 

Let $\mathbfit{F}$ be a tuple of nondeterministic facts $F_1,\dots,F_m\in chase(D,\Sigma)$ such that $F_i\cap I=\{h(\vartheta_i)\}$ for $1\le i\le m$. Let $k$ be an integer such that $\{F_1,\dots,F_m\}\subseteq chase_k(D,\Sigma)$. By definition of the chase procedure, we know  
\begin{equation*}
res(\mathbfit{F},\sigma,h)\in chase(D,\Sigma),
\end{equation*}
where 
\begin{equation*}
res(\mathbfit{F},\sigma,h)=\{h'(\vartheta_i):1\le i\le m\}\cup\bigcup_{i=1}^m F_i\setminus\{h(\vartheta_i)\},
\end{equation*}
and $h'$ is a substitution that extends $h$ in the way defined by $chase$.
So, $I$ must be a model of $res(\mathbfit{F},\sigma,h)$. Consequently, $h'(\alpha_j)\in I$ for some $1\le j\le n$, which is as desired.
\end{proof}

\subsection{Proof of Proposition~\ref{lem:chase_hom_prv}}

{\noindent\bf Proposition~\ref{lem:chase_hom_prv}.} 
Let $\Sigma$ be a set of DTGDs, let $D$ and $D'$ be databases, and let $C$ be a set of constants. If there exists a $C$-homomorphism $\tau$ from $D$ to $D'$, then there exists a $C$-homomorphism $\tau'\supseteq\tau$ from  $chase(D,\Sigma)$ to $chase(D',\Sigma)$.

\begin{proof}
Let $\tau$ be a $C$-homomorphism from $D$ to $D'$. Our task is to extend $\tau$ to a $C$-homomorphism $\tau'$ from $chase(D,\Sigma)$ to $chase(D',\Sigma)$. For convenience, let $\tau_0$ denote $\tau$. Let $k$ be any positive integer. Suppose $\tau_{k-1}$ is a $C$-homomorphism from $chase_{k-1}(D,\Sigma)$ to $chase_{k-1}(D',\Sigma)$. We need to prove that $\tau_{k-1}$ can be extended to a $C$-homomorphism $\tau_k$ from $chase_{k}(D,\Sigma)$ to $chase_{k}(D',\Sigma)$. If such a $\tau_k$ indeed exists, it is easy to verify that $\tau'=\bigcup_{k\ge 0}\tau_k$ is a $C$-homomorphism from $chase(D,\Sigma)$ to $chase(D',\Sigma)$.

So it remains to construct the $C$-homomorphism $\tau_k$. To do this, we first prove a property as follows: {\em Every DTGD in $\Sigma$ applicable to $chase_{k-1}(D,\Sigma)$ is also applicable to $chase_{k-1}(D',\Sigma)$.} Let $\sigma\in\Sigma$ be a DTGD which is applicable to $chase_{k-1}(D,\Sigma)$, and let $\alpha_1,\dots,\alpha_n$ list all the atoms in the body of $\sigma$. Then there must be a tuple $\mathbfit{F}$ of nondeterministic facts $F_1,\dots,F_n\in chase_{k-1}(D,\Sigma)$ and a substitution $h$ such that $h(\alpha_i)\in F_i$ for  $1\le i\le n$. By assumption, $\tau_{k-1}$ is a $C$-homomorphism from $chase_{k-1}(D,\Sigma)$ to $chase_{k-1}(D',\Sigma)$. Consequently, we have 
\begin{equation*}
\{\tau_{k-1}(F_1),\dots,\tau_{k-1}(F_n)\}\subseteq chase_{k-1}(D',\Sigma).
\end{equation*}
It is also clear that $\tau_{k-1}(h(\alpha_i))\in\tau_{k-1}(F_i)$ for $1\le i\le n$, which implies that $\sigma$ is applicable to $chase_{k-1}(D',\Sigma)$.

Let $\tau_k$ denote the smallest extension of $\tau_{k-1}$ such that, if $\sigma\in\Sigma$ is applicable to $chase_{k-1}(D,\Sigma)$ via some substitution $h$ and some tuple $\mathbfit{F}$ of nondeterministic facts in $chase_{k-1}(D,\Sigma)$, then $\tau_k(h'(x))=g'(x)$ for each existential variable $x$ in $\sigma$, where $h'$ (resp., $g'$) is the substitution extending $h$ (resp., $\tau_{k-1}\circ h$) and introduced in $res(\mathbfit{F},\sigma, h)$ (resp., $res(\tau_{k-1}(\mathbfit{F}),\sigma, \tau_{k-1}\circ h)$). It is easy to verify that 
{\small
$$\tau_k(res(\mathbfit{F},\sigma, h))=res(\tau_{k-1}(\mathbfit{F}),\sigma, \tau_{k-1}\circ h)\in chase_k(D',\Sigma),$$
}

\vspace{-.35cm}
\noindent which implies $\tau_k(chase_k(D,\Sigma))\subseteq chase_k(D',\Sigma)$. Consequently, $\tau_k$ is a $C$-homomorphism from $chase_k(D,\Sigma)$ to $chase_k(D',\Sigma)$. This thus completes the proof.
\end{proof}

\subsection{Proof of Proposition~\ref{prop:close_dh}}

{\noindent\bf Proposition~\ref{prop:close_dh}.} 
Every DTGD$[\mathsf{UCQ}]$-ontology is closed under database homomorphisms.

\begin{proof}
Let $\Sigma$ be a finite set of DTGDs. Let $\mathscr{D}$ and $\mathscr{Q}$ be a pair of schemas, $D$ a $\mathscr{D}$-database, and $q$ a Boolean $\mathscr{Q}$-UCQ. Suppose $D\cup\Sigma\vDash q$, and let $D'$ be a $\mathscr{D}$-database such that $D\rightarrow_C D'$, where $C$ denotes $const(q)$. To yield the desired proposition, it is sufficient to prove $D'\cup\Sigma\vDash q$. 

According to the completeness of nondeterministic chase, we have $chase(D,\Sigma)\vDash q$, and by the soundness of nondeterministic chase, it suffices to prove $chase(D',\Sigma)\vDash q$. Let $J$ be a model of $chase(D',\Sigma)$. We need to show that $J$ is a model of $q$. By Proposition~\ref{lem:chase_hom_prv}, we know that there is a $C$-homomorphism $\tau$ from $chase(D,\Sigma)$ to $chase(D',\Sigma)$. Let 
\begin{equation*}
I=\{\alpha\in atom(chase(D,\Sigma)): \tau(\alpha)\in J\},
\end{equation*}
where $atom(chase(D,\Sigma))$ denotes the set of all atoms that occur in some nondeterministic fact in $chase(D,\Sigma)$. Next we prove that $I$ is a model of $chase(D,\Sigma)$. 

Let $F$ be a nondeterministic fact in $chase(D,\Sigma)$. Clearly, we have $\tau(F)\in chase(D',\Sigma)$, which implies that $J$ is a model of $\tau(F)$, i.e., $\tau(F)\cap J\ne\emptyset$. By definition, we know that $F\cap I\ne\emptyset$, or equivalently, $I$ is a model of $F$. Due to the arbitrariness of $F$, we know that $I$ is indeed a model of $chase(D,\Sigma)$. Since $q$ is a consequence of $chase(D,\Sigma)$, we then obtain that $I$ is a model of $q$. On the other hand, by definition, $I$ is clearly $C$-homomorphic to $J$. Since $q$ is preserved under $C$-homomorphism, we immediately have that $J$ is also a model of $q$, which is that we need.
\end{proof}

\subsection{Proof of Proposition~\ref{prop:close_cs}}

{\noindent\bf Proposition~\ref{prop:close_cs}.} 
Every DTGD$[\mathsf{UCQ}]$-ontology is closed under constant substitutions.

\begin{proof}
Let $\Sigma$ be a finite set of DTGDs. Let $\mathscr{D}$ and $\mathscr{Q}$ be a pair of schemas, $D$ a $\mathscr{D}$-database, and $q$ a Boolean $\mathscr{Q}$-UCQ. Let $\tau$ be a constant substitution, i.e., a partial function from $\Delta$ to $\Delta$. By the soundness and completeness of nondeterministic chase, to yield Proposition~\ref{prop:close_cs}, it suffices to prove that $chase(D,\Sigma)\vDash q$ implies $chase(\tau(D),\Sigma)\vDash\tau(q)$.

Suppose we have $chase(D,\Sigma)\models q$. Our task is to prove $chase(\tau(D),\Sigma)\vDash\tau(q)$. It is easy to verify that $\tau$ is actually a homomorphism from $D$ to $\tau(D)$. According to Proposition~\ref{lem:chase_hom_prv}, there is a homomorphism $\tau'$ from $chase(D,\Sigma)$ to $chase(\tau(D),\Sigma)$ such that $\tau\subseteq\tau'$. Consequently, we have $$\tau'(chase(D,\Sigma))\subseteq chase(\tau(D),\Sigma).$$ As a consequence, to prove $chase(\tau(D),\Sigma)\vDash\tau(q)$, it suffices to prove $\tau'(chase(D,\Sigma))\vDash\tau(q)$. Let $J$ be a model of $\tau'(chase(D,\Sigma))$. Obviously, for every nondeterministic fact $F\in chase(D,\Sigma)$, there exists at least one disjunct, denoted $\alpha_F$, of $F$ such that $\tau'(\alpha_F)\in J$. Let $I$ denote the set consisting of $\alpha_F$ for all $F\in chase(D,\Sigma)$. It is easy to see that $I$ is a model of $chase(D,\Sigma)$. By assumption, we have $chase(D,\Sigma)\vDash q$, which implies that $I\models q$. W.l.o.g., we assume that $q$ is of the form $q_1\vee\cdots\vee q_n$ where each $q_i$ is a BCQ. Then there exists some $k\in\{1,\dots,n\}$ and a substitution $h$ such that $h([q_k])\subseteq I$, which implies that 
$$h([\tau(q_k)])=\tau(h([q_k]))=\tau'(h([q_k]))\subseteq\tau'(I)=J.$$
On the other hand, it is clear that $\tau(q)=\tau(q_1)\vee\cdots\vee\tau(q_n)$, which implies that $J$ is a model of $\tau(q)$. By the arbitrariness of $J$, we thus have $\tau'(chase(D,\Sigma))\vDash\tau(q)$ as desired.
\end{proof}

\subsection{Proof of Theorem~\ref{thm:char_dtgd}}

{\noindent\bf Theorem~\ref{thm:char_dtgd}.} 
A DED$[\mathsf{UCQ}]$-ontology is defined by a finite set of DTGDs iff it is closed under both database homomorphisms and constant substitutions.

\begin{proof}
The direction of ``only-if" immediately follows from Propositions~\ref{prop:close_dh} and~\ref{prop:close_cs}. It thus remains to consider the converse. 
Let $O$ be a DED$[\mathsf{UCQ}]$-ontology over some schema pair $(\mathscr{D},\mathscr{Q})$, and suppose it is closed under both database homomorphisms and constant substitutions. We need to prove that $O$ is defined by some finite set of DTGDs. Let $\Sigma$ be a finite set of DEDs that defines $O$. Let $\mathscr{S}$ be the schema of $\Sigma$. We introduce $\textit{Eq}$ as a fresh binary relation symbol, and $\textit{Dom}$ as a fresh unary relation symbol. Let $\Sigma_e$ be a set that consists of all the DTGDs defined in the following 1--3:
\begin{enumerate}
\item For each relation symbol $R\in\mathscr{S}$, if the arity of $R$ is $k$ and $1\le i\le k$, a DTGD of the form
\begin{equation}
R(x_1,\dots,x_k)\rightarrow\textit{Dom}(x_i),
\end{equation}
is used to collect terms (constants or nulls) occurring in the $i$-th position of $R$, 
where $k$ denotes the arity of $R$. 
\item The following DTGDs are employed to assert that $\textit{Eq}$ defines an equivalence relation in the intended model:
\begin{eqnarray}
\textit{Dom}(x)\rightarrow\textit{Eq}(x,x)\\
\textit{Eq}(x,y)\rightarrow\textit{Eq}(y,x)\\
\textit{Eq}(x,y)\wedge\textit{Eq}(y,z)\rightarrow\textit{Eq}(x,z)
\end{eqnarray}
\item 
For each relation symbol $R\in\mathscr{S}$, if the arity of $R$ is $k$, a DTGD of the form
\begin{equation}\label{eqn:eq_diffusion}
{\wedge}_{i=1}^k\textit{Eq}(x_i,y_i)\wedge{R}(x_1,\dots,x_k)\rightarrow{R}(y_1,\dots,y_k)
\end{equation}
is introduced to assure that all terms (constants or nulls) equivalent w.r.t. $\textit{Eq}$ will play the same role in $R$. 
\end{enumerate}


Let $\Sigma^\ast$ be a set of DTGDs obtained from $\Sigma$ by substituting $\textit{Eq}$ for all occurrences of the equality symbol $=$. Let $\Sigma^+$ denote $\Sigma^\ast\cup \Sigma_e$. Next we prove that, for all $\mathscr{D}$-databases $D$ and Boolean $\mathscr{Q}$-UCQs $q$, we have $D\cup\Sigma\vDash q$ iff $D\cup\Sigma^+\vDash q$. 

For the easy direction, let us assume $D\cup\Sigma^+\vDash q$. We need to prove $D\cup\Sigma\vDash q$. Let $I$ be a model of $D\cup\Sigma$, and let 
\begin{equation*}
J=I\cup\{(a,a):a\in adom(I)\}.
\end{equation*}
It is easy to see that $J$ is a model of $D\cup\Sigma^+$, which implies that $I$ is also a model of $q$. This then yileds $D\cup\Sigma\vDash q$.

For the converse, we assume $D\cup\Sigma\vDash q$. Let $J$ be a minimal model of $D\cup\Sigma^+$. Let $\sim$ denote the binary relation 
$$
\{(a,b):\textit{Eq}(a,b)\in J\}.
$$
Since $J$ is a model of $\Sigma_e$, $\sim$ must be an equivalence relation on $adom(J)$. For each $a\in adom(J)$, let $\widehat{a}$ denote the equivalence class of $a$ under $\sim$. We define $\tau$ as a function that maps each $c\in adom(D)$ to $\widehat{c}$. Since $O$ is closed under constant substitutions, from the assumption $D\cup\Sigma\vDash q$, we know that $\tau(D)\cup\Sigma\vDash\tau(q)$ holds. Let
\begin{equation*}
I=\{R(\widehat{a}_1,\dots,\widehat{a}_k): R(a_1,\dots,a_k)\in J\text{ and }R\in\mathscr{S}\}.
\end{equation*}
It is not difficult to see that $I$ is a model of $\tau(D)\cup\Sigma$, which implies that $I$ is a model of $\tau(q)$. From it, we infer that $J$ is a model of $q$, which yields  $D\cup\Sigma^+\vDash q$ as desired.
\end{proof}

\subsection{Proof of Lemma~\ref{lem:gcq2bcq}}

{\noindent\bf Lemma~\ref{lem:gcq2bcq}.} 
Let $\Sigma$ be a finite set of DTGDs, $D$ a database, and $q$ a BCQ of the form $\exists\mathbfit{x}\varphi(\mathbfit{x})$ where $\varphi$ is quantifier-free and $\mathbfit{x}$ is a tuple of length $k$ which lists all the free variables in $\varphi$. Then $D\cup\Sigma\vDash q$ iff there exists a finite set $T\subseteq term(chase(D,\Sigma))^k$ such that $chase(D,\Sigma)\vDash\bigvee_{\mathbfit{t}\in T}\varphi(\mathbfit{t})$.  

\begin{proof}
By the soundness and completeness of chase, we have $D\cup\Sigma\vDash q$ iff $chase(D,\Sigma)\vDash q$. By compactness, the latter holds iff there exists $n\ge 0$ such that $chase_n(D,\Sigma)\vDash q$. Let $T$ denote $term(chase_n(D,\Sigma))^k$. It is easy to check that $chase_n(D,\Sigma)\vDash q$ iff $chase_n(D,\Sigma)\vDash\bigvee_{\mathbfit{t}\in S}\varphi(\mathbfit{t})$. Since $T$ is always a finite set, the lemma follows.
\end{proof}

\subsection{Proof of Lemma~\ref{lem:cq2ucq}}

{\noindent\bf Lemma~\ref{lem:cq2ucq}.} 
Let $O$ and $O'$ be OMQA$[\mathsf{UCQ}]$-ontologies that admit query constructivity. If $O|_{\mathsf{CQ}}=O'|_{\mathsf{CQ}}$ then $O=O'$.

\begin{proof}
Suppose $O|_{\mathsf{CQ}}=O'|_{\mathsf{CQ}}$. We need to prove $O=O'$. Let $(D,q)\in O$. Due to the symmetry, it suffices to show that $(D,q)\in O'$. Suppose $q$ is of the form $q_1\vee\cdots\vee q_n$ where each $q_i$ is a BCQ. By the query constructivity of $O$, there must exist $i\in\{1,\dots,n\}$ such that $(O,q_i)\in O$, which implies $(O,q_i)\in O|_{\mathsf{CQ}}$ immediately. By assumption, we then obtain $(O,q_i)\in O'|_{\mathsf{CQ}}$. Consequently, we have $(O,q_i)\in O'$. On the other hand, it is clearly true that $q_i\vDash q$. Since $O'$ is closed under query implications, we thus conclude that $(O,q)\in O'$, which completes the proof.
\end{proof}

\subsection{Proof of Theorem~\ref{thm:char_ltgd}}

{\noindent\bf Theorem~\ref{thm:char_ltgd}.} 
Let $\mathscr{D}$ and $\mathscr{Q}$ be schemas. An OMQA$[\mathsf{CQ}]$-ontology $O$ over $(\mathscr{D},\mathscr{Q})$ is defined by a finite set of linear TGDs iff it admits both of the following properties:
\begin{enumerate}
\item (Data Constructivity) If $D$ and $D'$ are $\mathscr{D}$-databases and $q\in O(D\cup D')$ is prime w.r.t. $O(D\cup D')$, then we have either $q\in O(D)$ or $q\in O(D')$.
\item (NTA-recognizability of Queries) For every $\mathscr{D}$-database $D$ with a single fact, there exists an oblivious NTA which recognizes $O(D)$.
\end{enumerate}
\medskip

Before presenting the desired proof, we first prove some lemmas:


\begin{lem}\label{lem:single2general}
Let $O$ and $O'$ be OMQA$[\mathsf{CQ}]$-ontologies over the same schema pair, both admit database constructivity. Then $O=O'$ iff $O(D)=O'(D)$ for every database $D$ (over a proper schema) which consists of a single fact.
\end{lem}

\begin{proof}
The direction of ``only-if" is trivial. We only prove the converse. Suppose $(\mathscr{D},\mathscr{Q})$ is the schema pair of $O$ and $O'$. We first assume $O(D)=O'(D)$ for every $\mathscr{D}$-database $D$ which consists of a single fact. Let $D_0$ be a $\mathscr{D}$-database and $q$ a $\mathscr{Q}$-BCQ such that $q\in O(D_0)$. To yield the desire conclusion $O=O'$, due to the symmetry, it suffices to prove that $q\in O'(D_0)$. It is easy to see that there exist a sequence of BCQs $q_1,\dots, q_k\in O(D_0)$ such that  
\begin{enumerate}
\item $q_1\wedge\cdots\wedge q_k\vDash q$, and 
\item for each $i\in\{1,\dots,k\}$, $q_i$ is prime w.r.t. $O(D_0)$.
\end{enumerate} 
Since $O$ adimits database constructivity, we know that for each $i\in\{1,\dots,k\}$ there is a database $D_i\subseteq D_0$ with a single fact such that $q_i\in O(D_i)$, which implies $q_i\in O'(D_i)$ according to the assumption $O(D)=O'(D)$. Since $O'$ is closed under injective database homomorphisms, we then obtain $q_i\in O'(D_0)$ for each $i\in\{1,\dots,k\}$. Furthermore, as $O'$ is closed under both query conjunctions and query implications, we thus have $q\in O'(D_0)$ as desired. 
\end{proof}


\begin{lem}\label{lem:data_constructivity}
Let $\Sigma$ be a set of linear TGDs, $D$ and $D'$ be databases, $\mathcal{C}$ be the class of BCQs $p$ over a given schema $\mathscr{Q}$ such that $D\cup D'\cup\Sigma\vDash p$. Let $q\in\mathcal{C}$ be a BCQ that is prime w.r.t. $\mathcal{C}$. Then we have either $D\cup\Sigma\vDash q$ or $D'\cup\Sigma\vDash q$. 
\end{lem}

\begin{proof}
Let $C$ denote $adom(D)\cup adom(D')$. To simplify the proof, w.l.o.g., we assume that 
\begin{equation}\label{eqn:disjointness}
term(chase(D,\Sigma))\cap term(chase(D',\Sigma))\subseteq C.
\end{equation}
Note that, if this is not true, by properly renaming nulls one can find an isomorphic copy of $chase(D',\Sigma)$ to make it true. 

To prove the above lemma, we need some notation and a property. Given instances $I$ and $J$ over the same schema and a set $X$ of constants, we write $I\leftrightarrow_X J$ if both $I\rightarrow_X J$ and $J\rightarrow_X I$ hold. Next let us prove the desired property:

\medskip
{\noindent\em Claim.} 
 $chase(D\cup D',\Sigma)\leftrightarrow_C chase(D,\Sigma)\cup chase(D',\Sigma)$.
 
 \begin{proof}
 It suffices to show that, for all $i\ge 0$, we have
 \begin{equation}\label{eqn:dc_statement}
 chase_i(D\cup D',\Sigma)\leftrightarrow_C chase_i(D,\Sigma)\cup chase_i(D',\Sigma).
 \end{equation}
 
 It is trivial for the case where $i=0$. For the case where $i>0$, let us assume as inductive hypothesis that 
 \begin{equation}\label{eqn:dc_indhypothesis}
 chase_{i-1}(D\cup D',\Sigma)\leftrightarrow_C 
 chase_{i-1}(D,\Sigma)\cup chase_{i-1}(D',\Sigma).
\end{equation}
We need to prove the property presented in~(\ref{eqn:dc_statement}). The direction of ``$\leftarrow_C$" immediately follows from the monotonicity of the chase operator. So, it remain to prove the converse, i.e., 
 \begin{equation}
 chase_{i}(D\cup D',\Sigma) \rightarrow_C  chase_{i}(D,\Sigma)\cup chase_{i}(D',\Sigma).
\end{equation}
Let $\tau$ be a $C$-homomorphism from $chase_{i-1}(D\cup D',\Sigma)$ to $chase_{i-1}(D,\Sigma)\cup chase_{i-1}(D',\Sigma)$.
Let $\sigma\in\Sigma$ and $h$ be a substitution such that $\sigma$ is applicable to $chase_{i-1}(D\cup D',\Sigma)$ via $h$. Let $h'$ be a substitution extending $h$ and defined according to the chase procedure. Let $\alpha$ be the only atom in the body of $\sigma$. Then we know that $h(\alpha)\in chase_{i-1}(D\cup D',\Sigma)$. Clearly, either $\tau(h(\alpha))\in chase_{i-1}(D,\Sigma)$ or $\tau(h(\alpha))\in chase_{i-1}(D',\Sigma)$ must hold. Due to the symmetry, we only consider the former, that is, $\tau(h(\alpha))\in chase_{i-1}(D,\Sigma)$. Let $g$ denote the substitution $\tau\circ h$. Then $\sigma$ must be applicable to $chase_{i-1}(D,\Sigma)$ via $g$. Let $g'$ be a substitution extending $g$ and defined by the chase procedure on $D$ and $\Sigma$. Let $\tau'$ be a function extending $\tau$ by mapping $h'(v)$ to $g'(v)$ for each existential variable $v$ in $\sigma$. Let $\tau^\ast$ be the function defined by repeating such a procedure for all TGDs $\sigma\in\Sigma$. It is not difficult to see that $\tau^\ast$ is a $C$-homomorphism from $chase_{i}(D\cup D',\Sigma)$ to $chase_{i}(D,\Sigma)\cup chase_{i}(D',\Sigma)$, which completes the proof for the desired direction.
 \end{proof}
 
Now we are in the position to prove the desired lemma. Since $q\in\mathcal{C}$, we know $D\cup D'\cup\Sigma\vDash q$. By the soundness and completeness of the chase, we have $chase(D\cup D',\Sigma)\models q$, and it is sufficient to show that either $chase(D,\Sigma)\models q$ or $chase(D',\Sigma)\models q$ holds. By the above claim, we know $chase(D,\Sigma)\cup chase(D',\Sigma)\models q$, which means that there exists a $C$-homomorphism $h$ from $q$ to $chase(D,\Sigma)\cup chase(D',\Sigma)$. Since $q$ is inseparable, for every (existential) variable $v$ in $q$, $h(v)$ must not be a constant. Otherwise, let $q'$ be the BCQ obtained from $q$ by substituting $h(v)$ for $v$. It is obvious that $chase(D,\Sigma)\cup chase(D',\Sigma)\models q'$. By Claim, we have $D\cup D'\cup\Sigma\vDash q'$, which  contradicts with the primality of $q$. Thus, $h$ maps each variable in $q$ to a null. By Assumption~(\ref{eqn:disjointness}) and the inseparability of $q$, we conclude that either $h(q)\subseteq chase(D,\Sigma)$ or $h(q)\subseteq chase(D',\Sigma)$ should be true, which then yields the desired conclusion.
%
\end{proof}

With the above lemmas, we are now able to prove Theorem~\ref{thm:char_ltgd}.

\begin{proof}[Proof of Theorem~\ref{thm:char_ltgd}]
``Only-if": Suppose $O$ is defined by a finite set $\Sigma$ of linear TGDs. We need to show that $O$ admits Properties 1 and 2. Property 1 can be proven by a careful induction on the chase. Below we prove that $O$ admits Property 2. 


Let $D$ be a $\mathscr{D}$-database with a single fact. Now let us construct an oblivious NTA that recognizes $O(D)$. Let $\mathscr{S}$ denote the schema of $\Sigma$. Let $k$ be the maximum arity of relation symbols in $\mathscr{S}$. Let $\mathcal{V}$ be the set that consists of pairwise distinct variables $x_1,\dots,x_{2k}$. Let $At$ be the set of all atoms built upon relation symbols from $\mathscr{S}$ and terms from $adom(D)\cup\mathcal{V}$. We introduce $\lceil\log_2 (|At|+2)\rceil$ fresh variables, and let $\mathcal{V}_0$ denote the set that consists of these variables. Let $\iota$ be an injective function from $At$ to $2^{\mathcal{V}_0}\setminus\{\emptyset,\mathcal{V}_0\}$. Thus, every atom in $At$ can be encoded by a set of variables in $\mathcal{V}_0$. 

With the above assumptions, we are now able to define the NTA. Let $S=At\cup\{\diamond\}$ be the set of states, and let $F=\{\diamond\}$ be the set that consists of the unique finite state $\diamond$. Furthermore, let $\mathcal{L}$ be a label set which consists of 
\begin{enumerate}
\item $(term(\alpha),\{\alpha\})$ for each $\alpha\in At$ which is a $\mathscr{Q}$-atom;
\item $(term(\alpha)\cup\iota(\alpha),\emptyset)$ for each $\alpha\in At$;
\item $(adom(D)\cup\mathcal{V}_0,\emptyset)$. 
\end{enumerate}
For convenience, we define a function $\lambda:\mathcal{L}\rightarrow S$ that maps each label $\ell\in\mathcal{L}$ of the form 1 or 2 to the atom (state) $\alpha$, and maps the label of form 3 to the final state $\diamond$. Clearly, $\lambda$ is well-defined. Let $\Omega$ be a ranked input alphabet which consists of ordered pairs $(\ell,m)$ for all $\ell\in\mathcal{L}$ and all $0\le m\le |At|$, where each $(\ell, m)$ is used as an $m$-ary input symbol.

%

Furthermore, let $\Theta$ be a set consisting of
\begin{enumerate}
\item $((\ell,1), \alpha, \diamond)$ if $\lambda(\ell)=\diamond$ and $D=\{\alpha\}$;
\item $((\ell,m), (\alpha_1,\dots,\alpha_m), \alpha)\in\Delta$ if $\lambda(\ell)=\alpha$, $0\le m\le |At|$, $\alpha,\alpha_1,\dots,\alpha_m\in At$ and for $1\le i\le m$, $\{\alpha\}\cup\Sigma\vDash\exists\mathbfit{x}_i\alpha_i$, where $\mathbfit{x}_i$ is a tuple consisting of all the variables that occur in $\alpha_i$ but not in $\alpha$.
\end{enumerate} 

Let $\mathcal{A}=(S,F,\Omega,\Theta)$. Since $\lambda$ is a well-defined function, we know that $\mathcal{A}$ is an oblivious NTA. Next we show that $\mathcal{A}$ recognizes the class $O(D)$. By the definition of $\mathcal{A}$, it is easy to see that every $\mathscr{Q}$-BCQ accepted by $\mathcal{A}$ belongs to $O(D)$. 

Conversely, let $q\in O(D)$. We need to prove that $\mathcal{A}$ accepts $q$. 
Let $\mathfrak{D}$ be a labeled tree constructed as follows:
\begin{enumerate}
\item Create the root $r$ with the label $L(r)=(adom(D),D)$;
\item 
For each node $t$ already in $\mathfrak{D}$, if there is an atom $\alpha\in At$ such that $L^2(t)\cup\Sigma\vDash\exists\mathbfit{x}\,\alpha$, 
then create a child $t'$ for $t$ and let $L(t')\!=\!(term(\alpha^\ast),\{\alpha^\ast\})$, where $\alpha^\ast$ is obtained from $\alpha$ by substituting fresh variables for variables in $\mathbfit{x}$.

\end{enumerate}
Let $atom(\mathfrak{D})$ be the set of all atoms appearing in $\mathfrak{D}$. Let $C$ denote $const(q)$. By definition we know that $chase(D,\Sigma)$ is $C$-isomorphic to a subset of $atom(\mathfrak{D})$. 
%
%
Next we prove the following property:

\medskip
{\em\noindent Claim 1.} $[q]$ is $C$-isomorphic to a subset of $atom(\mathfrak{D})$.

\begin{proof}
By regarding variables in $\mathfrak{D}$ as nulls, we regard $atom(\mathfrak{D})$ as an instance. Since we have  $q\in O(D)$, according to the previous conclusion we conclude that $atom(\mathfrak{D})\models q$, i.e., there is a $C$-homomorphism $h$ from $[q]$ to $atom(\mathfrak{D})$. Let $\mathfrak{D}_q$ be a minimal subtree of $\mathfrak{D}$ that covers $h([q])$ and the root $r$. Now we let $\mathfrak{D}'_q$ be a labeled tree obtained from $\mathfrak{D}_q$ by
\begin{quote}
for each node $t$ in $\mathfrak{D}_q$ with $L(t)=(term(h(\alpha)),\{h(\alpha)\})$ for some $\alpha\in[q]$, adding a fresh node $t'$ as a child of $t$, and setting the label $L(t')=(term(\alpha),\{\alpha\})$.
\end{quote}
Again, let $atom(\mathfrak{D}'_q)$ denote the set of all atoms appearing in $\mathfrak{D}'_q$. Clearly, it holds that $[q]\subseteq atom(\mathfrak{D}'_q)$. On the other hand, given any atom $\alpha$, let $\mathbfit{x}$ be a tuple that consists of all the variables appearing in $\alpha$ but not in $h(\alpha)$. Clearly, we always have that $h(\alpha)\vDash\exists\mathbfit{x}\alpha$. According to the construction of $\mathfrak{D}$, it is then easy to see that there must be an injective renaming function $\tau$ on variables (i.e., $\tau$ is a partial function on $\Delta_{\mathrm{v}}$) such that $\tau(\mathfrak{D}'_q)$ is indeed a subtree of $\mathfrak{D}$, which thus yields the desired claim.
\end{proof}

With Claim 1, there must be a subset, denoted $Q$, of $atom(\mathfrak{D})$ such that $[q]$ is $C$-isomorphic to $Q$. Let $\mathfrak{T}_q=(V,E,r,L)$ be a minimal subtree of $\mathfrak{D}$ that covers $Q$ and the root of $\mathfrak{D}$. Let $\mathfrak{R}_q$ be the labeled tree $(V,E,r,L_0)$, where $L_0$ is defined as follows: 
\begin{enumerate}
\item for the root $r$, let $L_0(r)=(adom(D)\cup\mathcal{V}_0,\emptyset)$;
\item for each node $t\in V$ with the label $L(t)=(term(\alpha),\{\alpha\})$, let $L_0(t)=(term(\alpha)\cup\iota(\alpha),\emptyset)$ if $\alpha\not\in Q$, and $L_0(t)=L(t)$ otherwise.  
\end{enumerate}
Clearly, $\mathfrak{R}_q$ is a finite and linear tree representation of $q$. By the technique mentioned in the last subsection, such a tree can be naturally encoded by an $\Omega$-ranked tree, which can be easily showed to be accepted by $\mathcal{A}$ by a routine check.
%

\medskip
``If": Suppose $O$ admits Properties 1 and 2. We need to prove that there is a finite set of linear TGDs that defines $O$. Let $D$ be a $\mathscr{D}$-database. We first prove the desired conclusion for a simple case.

\medskip
{\noindent\em Claim 2.} If $|D|=1$ then there is a finite set $\Sigma$ of linear TGDs such that, for every $\mathscr{Q}$-BCQ $q$, $D\cup\Sigma\vDash q$ iff $q\in O(D)$.

\begin{proof}
Let $D$ be a $\mathscr{D}$-database which consists of a single fact. By Property 2, there exists an oblivious NTA $\mathcal{A}=(S,F,\Omega,\Theta)$ that recognizes $O(D)$. Based on $\mathcal{A}$, we will construct a finite set of linear TGDs that defines $O$. 
%
%

Before proceeding, we first define some notations. For each constant $c$, we introduce  $v_c$ as a fresh variable that has never been used in $\Omega$. Given an atom $\alpha$, let $\widehat{\alpha}$ denote the atom obtained from $\alpha$ by substituting $v_c$ for each constant $c$. Given a tuple $\mathbfit{c}$ of constants $c_1,\dots,c_k$, let $\widehat{\mathbfit{c}}$ denote the tuple of variables $v_{c_1},\dots,v_{c_k}$. Let $\mathcal{L}$ be the set of labels used in $\Omega$. According to the assumption made in Subsection ``Automata That Accept BCQs", we know that each label in $\mathcal{L}$ is an ordered pair $\omega=(X,\Gamma)$ where $X$ is a set of variables and $\Gamma$ a set of atoms. Let $|\omega|$ denote the number of variables in $X$. For each pair $(s,\omega)\in S\times\mathcal{L}$, we introduce a fresh relation symbol $T_{s,\omega}$ of arity $|\omega|$.

Now we let $\Sigma_\mathcal{A}$ be a set of linear TGDs defined as follows: 
\begin{enumerate}
\item Suppose $\alpha$ is the only fact in $D$. If $((\omega,\ell), \mathbfit{s},s')\in\Theta$ and $s'\in F$, then let $\Sigma_\mathcal{A}$ contain the TGD
\begin{equation}
\widehat{\alpha}\rightarrow\exists \mathbfit{x}\, T_{s',\omega}(\widehat{\mathbfit{c}},\mathbfit{x})
\end{equation} 
where $\mathbfit{x}$ (resp., $\mathbfit{c}$) is the tuple of variables (resp., constants) appearing in $\omega$. 

\item If $\{((\omega_1,\ell_1),\mathbfit{s}_1,s_1'),((\omega_2,\ell_2),\mathbfit{s}_2,s'_2)\}\subseteq\Theta$ and $s_1'\in\mathbfit{s}_2$,\footnote{To simplify the presentation, given a tuple $\mathbfit{a}$, by $a\in\mathbfit{a}$ we mean that $a$ is a component of $\mathbfit{a}$.} then let $\Sigma_\mathcal{A}$ contain the TGD
\begin{equation}
T_{s_2',\omega_2}(\widehat{\mathbfit{c}}_2,\mathbfit{x},\mathbfit{y})\rightarrow\exists\mathbfit{z}\,T_{s_1',\omega_1}(\widehat{\mathbfit{c}}_1,\mathbfit{x},\mathbfit{z})
\end{equation}
where, for $i=1$ or $2$, $\mathbfit{c}_i$ denotes the tuple of constants occurring in $\omega_i$; $\mathbfit{x}$ is a tuple consisting of all the variables that occur in both $\omega_1$ and $\omega_2$; and $\mathbfit{y}$ (resp., $\mathbfit{z}$) is a tuple consisting of all the variables that occur in $\omega_2$ but not in $\omega_1$ (resp., in $\omega_1$ but not in $\omega_2$).

\item If $(X,\{\alpha\})\in\mathcal{L}$ and $s\in S$, let $\Sigma_\mathcal{A}$ contain the TGD
\begin{equation}
T_{s,\omega}(\widehat{\mathbfit{c}},\mathbfit{x})\rightarrow\widehat{\alpha}
\end{equation}
where $\mathbfit{c}$ (resp., $\mathbfit{x}$) is a tuple consisting of all the constants (resp., variables) in $X$, and $\alpha$ is an atom.
\end{enumerate}
Note that all the variables used above are assumed to be ordered in a fixed way, and variables in a tuple will follow this order.  


Clearly, $\Sigma_{\mathcal{A}}$ is a finite set of linear TGDs. By a careful induction on the chase, it is not difficult to prove that $D\cup\Sigma_{\mathcal{A}}\vDash q$ iff $\mathcal{A}$ accepts $q$ for every $\mathscr{Q}$-BCQ $q$, which yields the claim.
\end{proof}

Next we prove there is a finite set of linear TGDs defines $O$. Obviously, up to isomorphism, there are only a finite number of $\mathscr{D}$-databases with a single fact. Let $D_1,\dots,D_n$ be a complete list of such databases. By Claim 2, for each $i\in\{1,\dots,n\}$, there exists a finite set $\Sigma_i$ of linear TGDs such that, for every $\mathcal{Q}$-BCQ $q$, we have $D\cup\Sigma_i\vDash q$ iff $q\in O(D_i)$.

For each relation symbol $R$ occurring in $\Sigma_i$ for some $i$, we introduce a fresh relation symbol $R_i$ of the same arity as $R$. Let $\Sigma_i'$ be the TGD set obtained from $\Sigma_i$ by substituting $R_i$ for all occurrences of $R$. Now we know that the schemas of $\Sigma_i'$, $1\le i\le n$, are pairwise disjoint. 

Let $\Gamma_1$ be a set that consists of the TGD
\begin{equation}
R(\mathbfit{x})\rightarrow R_i(\mathbfit{x})
\end{equation} 
for all $i=1,\dots,n$ and $R\in\mathscr{D}$, where $\mathbfit{x}$ is a variable tuple of a proper length. Let $\Gamma_2$ be a set consisting of the TGD
\begin{equation}
R_i(\mathbfit{x})\rightarrow R(\mathbfit{x})
\end{equation} 
for all $i=1,\dots,n$ and $R\in\mathscr{Q}$, where $\mathbfit{x}$ is a variable tuple of a proper length. Furthermore, let
$$\Sigma=\Gamma_1\cup\Gamma_2\cup\bigcup_{i=1}^n\Sigma_i.$$ 
To prove that $\Sigma$ defines $O$, according to Lemma~\ref{lem:single2general}, it suffices to prove that, given any $\mathscr{D}$-databases $D$ with a single fact and any $\mathscr{Q}$-BCQ $q$, we have $D\cup\Sigma\vDash q$ iff $q\in O(D)$, which can be done by a routine check. 
%
%
\end{proof}

}

\end{document}